\documentclass[letterpaper,journal]{IEEEtran}
\usepackage{amsmath,amsfonts}
\usepackage{array}
\usepackage[caption=false,font=normalsize,labelfont=sf,textfont=sf]{subfig}
\usepackage{textcomp}
\usepackage{stfloats}
\usepackage{url}
\usepackage{verbatim}
\usepackage{graphicx}
\usepackage{cite}
\usepackage{amsmath,amssymb,amsfonts}
\def\BibTeX{{\rm B\kern-.05em{\sc i\kern-.025em b}\kern-.08em
    T\kern-.1667em\lower.7ex\hbox{E}\kern-.125emX}}

\usepackage{mathrsfs}
\usepackage{multicol}
\usepackage{wrapfig}
\usepackage[utf8]{inputenc}
\usepackage{amsthm}
\usepackage[shortlabels]{enumitem}
\usepackage{here}
\usepackage{multirow}
\usepackage{stmaryrd}
\usepackage{version}

\usepackage{lscape} 
\usepackage{comment}
\usepackage{listings}
\usepackage{verbatim}
\usepackage{xcolor}
\usepackage{enumitem}
\usepackage{cases}
\usepackage{xspace}
\usepackage{siunitx}
\usepackage{tablefootnote}
\usepackage{booktabs}

\usepackage{bm}
\usepackage{bbm}
\usepackage{algorithm2e}
\RestyleAlgo{ruled}
\SetAlCapNameFnt{\small}
\SetAlCapFnt{\small}
\usepackage{algpseudocode}
\newcommand{\1}{\mbox{\fontencoding{U}\fontfamily{bbold}\selectfont1}}
\newcommand{\0}{\mbox{\fontencoding{U}\fontfamily{bbold}\selectfont0}}

\DeclareMathAlphabet{\mathcal}{OMS}{cmsy}{m}{n}

\newtheorem{theorem}{Theorem}
\newtheorem{corollary}[theorem]{Corollary}
\newtheorem{lemma}[theorem]{Lemma}

\newtheorem{prop}[theorem]{Proposition}
\newtheorem{defn}[theorem]{Definition}

\usepackage{hyperref}%
\hypersetup{colorlinks,allcolors=black}

\renewcommand{\natural}{{\mathbb{N}}}

\newcommand{\real}{\ensuremath{\mathbb{R}}}

\renewcommand{\theenumi}{(\roman{enumi})}

\DeclareSymbolFont{bbold}{U}{bbold}{m}{n}
\DeclareSymbolFontAlphabet{\mathbbold}{bbold}

\newcommand{\E}{\mathbb{E}}

\newcommand{\kemeny}{\mathcal{K}}
\newcommand{\liftK}{\mathcal{K}^\textup{lift}}

\renewcommand{\Pr}{\mathbb{P}}

\newcommand{\Pnl}{P}
\newcommand{\pnl}{p}
\newcommand{\pinl}{\pi}
\newcommand{\Pinl}{\Pi}
\newcommand{\Qnl}{Q}
\newcommand{\qnl}{q}
\newcommand{\xnl}{\mathbf{x}}
\newcommand{\Xnl}{X}
\newcommand{\Wnl}{W}
\newcommand{\wnl}{w}
\newcommand{\nnl}{n}
\newcommand{\Snl}{\mathcal{S}}

\newcommand{\Pl}{P^\textup{lift}}
\newcommand{\pl}{p^\textup{lift}}
\newcommand{\pil}{\pi^\textup{lift}}
\newcommand{\Pil}{\Pi^\textup{lift}}
\newcommand{\Ql}{Q^\textup{lift}}
\newcommand{\ql}{q^\textup{lift}}
\newcommand{\Xl}{X^\textup{lift}}
\newcommand{\xl}{\mathbf{x}^\textup{lift}}
\newcommand{\Wl}{W^\textup{lift}}
\newcommand{\wl}{w^\textup{lift}}
\renewcommand{\nl}{n^\textup{lift}}
\newcommand{\Sl}{\mathcal{S}^\textup{lift}}

\newcommand{\Plp}{P^\textup{lift,prod}}
\newcommand{\plp}{p^\textup{lift,prod}}
\newcommand{\pilp}{\pi^\textup{lift,prod}}
\newcommand{\Slp}{\mathcal{S}^\textup{lift,prod}}

\newcommand{\Qlt}{\tilde{Q}^\textup{lift}}

\usepackage{fancyhdr}
\fancypagestyle{preprint_notice}{
  \fancyhf{}
  
  \fancyhead[C]{\scriptsize This work has been submitted to the IEEE for possible publication. \\ Copyright may be transferred without notice, after which this version may no longer be accessible.}
}

\begin{document}

\title{Deceptive Stochastic Patrolling via Markov Chain Lifting}

\author{Yohan John*, Gilberto D\'iaz-Garc\'ia*, Jason R. Marden, and Francesco Bullo
\thanks{This work was supported by the Air Force Office of Scientific Research under grants FA9550-21-1-0203 and FA9550-25-1-0245.}
\thanks{The authors are with the Center for Control, Dynamical Systems, and Computation, UC Santa Barbara, Santa Barbara, CA 93106-5070, USA. * denotes equal contribution. Email: \{{\tt\small yohanjohn}, {\tt\small gdiaz-garcia}, {\tt\small jrmarden}, {\tt\small bullo}\}{\tt\small @ucsb.edu}}%
}

\maketitle
\thispagestyle{preprint_notice}

\begin{abstract}
In this paper we propose lifted Markov Chains (MCs) as a new paradigm for deriving patrol strategies for mobile agents on an environment represented as a graph. Lifted MCs operate on a state space that can be larger than the set of nodes of the graph, and a projection operation maps the MC state to the corresponding node of the graph. We prove bounds on the performance improvement of lifted MCs over non-lifted MCs, where the performance is measured by the Kemeny constant, and show a corresponding monotonicity for the Stackelberg game capture probability and return-time entropy. These bounds yield a result of independent interest that relates the Kemeny constant and the MC conductance for general MCs. We also provide a tractable method for optimizing lifted MCs on general graphs with edge weights encoding travel times. Simulation results on randomly generated and real-world graphs verify that lifted MCs achieve substantial improvements, particularly on sparse graphs. 
\end{abstract}


\section{Introduction}
In this paper we design patrol strategies for mobile agents in a known environment. We focus on stochastic patrol strategies to prevent exploitation by an adversary. There is a substantial body of literature on this topic~\cite{NB:22}. The predominant approach is to model the patrol environment as a graph where nodes represent locations of interest and edges represent paths that an agent can traverse. Markov Chains (MCs) are a natural choice for encoding stochastic patrol strategies on a graph, and they allow our strategies to be optimized offline and to impose minimal memory and compute requirements for online execution. Reference~\cite{XD-FB:20a} contains a review of several prior works that adopt MC-based stochastic patrolling. In these works, there is a bijection between MC states and graph nodes and thus state transition probabilities correspond to edge traversal probabilities. In this paper we propose the application of MC lifting, a classical technique that increases the size of the state space, to stochastic patrolling. Having multiple MC states correspond to the same node of the graph enables lifted MCs to capture more complex patrol behavior. We illustrate the concept and the associated performance improvement with the following example that we will return to throughout. 

\subsection{Running Example: Corridor Environment}

\begin{figure}
\centering
\includegraphics[width=\linewidth]{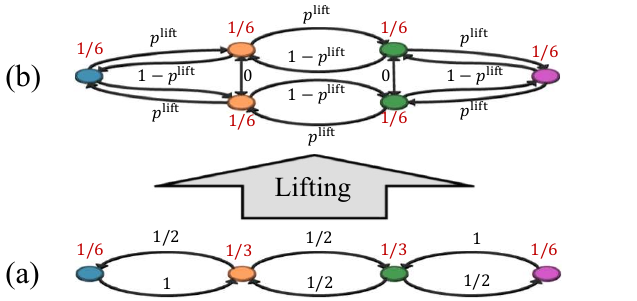}
\caption{An example of lifting for the corridor environment with four nodes. Figure~(a) shows the optimized non-lifted MC $\Pnl$, and Figure~(b) shows a particular class of lifted MCs $\Pl$ parametrized by $\pl\in [0,1]$ where virtual copies have been made of the two central nodes of the graph. The transition probabilities are shown in black, and the stationary probabilities are shown in red. Self-loops have probability zero and are hidden.}
\label{fig:motivating_example}
\end{figure}

We adopt the setting of \cite{RP-PA-FB:14b} in which the performance of an MC patrol strategy is measured via the Kemeny constant (defined in Section~\ref{sec:metrics}) and the MC stationary distribution is constrained to encode priority among the patrol locations. Inspired by~\cite{PD-SH-RMN:00}, we consider a corridor patrol environment modeled as a path graph to demonstrate the advantages of lifting. We set the stationary distribution proportional to the degree of each node. For simplicity, consider a path graph with $\nnl = 4$ nodes and stationary distribution $\pinl = \begin{bmatrix} 1/6 & 1/3 & 1/3 & 1/6 \end{bmatrix}^\top$. Using a constrained nonlinear optimization solver we identify the non-lifted MC $\Pnl$ shown in Fig.~\ref{fig:motivating_example}(a) that minimizes the Kemeny constant subject to the graph and stationary distribution constraints. It achieves a Kemeny constant $\kemeny(\Pnl) = 25/6$. Now consider creating virtual copies of the two central nodes of the graph. This affords more parameters to the lifted MC. Consider lifted MCs $\Pl$ of the form shown in Fig.~\ref{fig:motivating_example}(b). Clearly, $\Pl$ meets the graph constraints. For any value of $\pl\in[0,1]$, $\Pl$ can also be verified to meet the stationary distribution constraint in the sense that summing stationary probabilities of virtual copies yields $\pinl$. Fig.~\ref{fig:p_vs_M} shows the Kemeny constant $\kemeny(\Pl)$ as a function of $\pl$. For $\pl < 0.13$ and $\pl > 0.87$, $\Pl$ achieves a lower Kemeny constant than $\Pnl$. Consider the special case of $\Pl$ of the form shown in Fig.~\ref{fig:motivating_example}(b) where $\pl \in \{0,1\}$. This achieves a Kemeny constant $\kemeny(\Pl) = 7/2$, a meaningful improvement over $\kemeny(\Pnl)$. The improvement grows as the number of nodes $\nnl$ in the path graph increases. For arbitrary path graphs with $\nnl \in \natural_{>1}$ and a desired stationary distribution proportional to the degree of each node, it can be shown that the lifted MC achieves linear scaling while the non-lifted MC scales quadratically. The first two columns of Table~\ref{tab:running_example} summarize these findings. This observation motivates the study of lifted MCs as patrol strategies. 

\begin{figure}
\centering
\includegraphics[width=\linewidth]{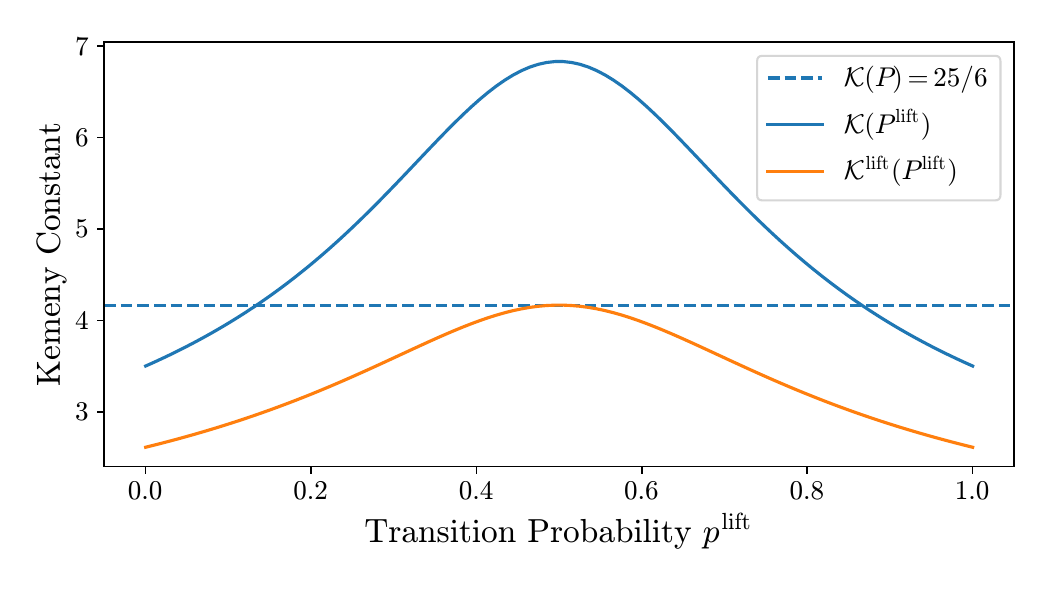}
\caption{A comparison of the performance on the corridor environment of the MC $\Pnl$ shown in Fig.~\ref{fig:motivating_example}(a) and the lifted MC $\Pl$ shown in Fig.~\ref{fig:motivating_example}(b). The dashed blue line shows the Kemeny constant of $\Pnl$, the solid blue shows the Kemeny constant of $\Pl$ as a function of $\pl$, and the solid orange shows the lifted Kemeny constant (defined in Section~\ref{sec:lifted_metrics}) of $\Pl$ as a function of $\pl$. For extreme values of $\pl$ the lifted strategy outperforms the non-lifted strategy by the Kemeny constant metric and it achieves better or equal performance for all values of $\pl$ by the lifted Kemeny constant metric.}
\label{fig:p_vs_M}
\end{figure}

\setlength{\tabcolsep}{4pt}
\begin{table}[ht]
    \centering
    \caption{Summary of Results for Corridor Environment}
    \label{tab:running_example}
    \renewcommand{\arraystretch}{1.75}
    \begin{tabular}{|c|c|c|c|c|}
        \hline
        & $\kemeny(\Pnl)$ & $\kemeny(\Pl)$ & $\liftK(\Pl)$ & $[2\Phi(\Pnl)]^{-1}$ \\
        \hline
        $\nnl = 4$ & $\frac{25}{6}$ & $\frac{7}{2}$ & $\frac{47}{18}$ & $\frac{3}{4}$ \\
        \hline
        $\nnl \in \natural_{>1}$ & $\frac{2\nnl^2-4\nnl+9}{6}$ & $\nnl-\frac{1}{2}$ & $\frac{4\nnl^2-5\nnl+3}{6(\nnl-1)}$ & $\frac{(\nnl-1)^2 - (\nnl \bmod 2)}{4(\nnl-1)}$ \\
        \hline
    \end{tabular}
\end{table}
\setlength{\tabcolsep}{6pt}

\subsection{Related Work}
MC-based stochastic patrolling is an active research area. Reference~\cite{RP-PA-FB:14b} is the foundational paper in which the Kemeny constant was proposed as an appropriate metric for the performance of MC patrol strategies. Recent work includes~\cite{XD-WW-RY:26} in which MC patrol strategies are designed to maximize trajectory entropy subject to a Kemeny constant constraint. Multi-agent MC-based patrolling has been explored in~\cite{BA-DSB:12, GDG-FB-JRM:22v}. However, in all of these works, the authors include a reversibility constraint on the MCs. Reversibility has been shown to limit the performance of MC-based patrol strategies~\cite{CPCF-AZ-BFH:26}. In this paper we avoid imposing reversibility and provide a tractable procedure for optimizing lifted MCs. 

Several non-MC approaches have been proposed to handle multi-agent patrolling of large environments. In~\cite{ABA-SS-SLS:25}, the authors develop an algorithm for identifying the number of agents required to achieve desired nodal refresh times. In~\cite{XL-CGC:15}, the authors formulate multi-agent monitoring of a continuous two-dimensional space as an optimal control problem. The challenging multi-agent stochastic patrolling setting has seen less work. Reference~\cite{DP-RPR:16} is a notable exception in which agents use Bayesian decision rules to minimize refresh time. In this paper, we focus on single-agent stochastic patrol strategies for small graphs with the understanding that large graphs can be partitioned and each subgraph assigned a single agent~\cite{DP-RPR:10}. 

Two recent works also explore MC lifting. In~\cite{SA-AS-FT:21}, the authors analyze the possible speed up in MC mixing time via lifting. They prove that the mixing time is lower bounded by a quantity proportional to the inverse of the conductance. Our lower bound in Theorem~\ref{thm:bounds} is an analogous result for the Kemeny constant. In the context of stochastic patrolling, the authors of~\cite{AN-MTG-MW-BR-MR-TPM:24} propose a technique similar to lifting for a generalization of the Stackelberg game defined in Section~\ref{sec:metrics}. This paper differs from~\cite{AN-MTG-MW-BR-MR-TPM:24} in three main ways: (i) we take a given non-lifted MC and optimize a lifting that preserves the stationary distribution, (ii) we consider the Kemeny constant and return-time entropy in addition to the Stackelberg game capture probability, and (iii) we consider an adversary that is unaware of the lifted space. We show that lifting in this setting can lead to deception of the adversary. 

This aspect relates our work to a growing body of literature on deception and information hiding against observers. In~\cite{PM-SS:17}, an agent performs deceptive path planning by hiding its target from an observer until as late as possible. In~\cite{NT-NA:17}, a multi-agent patrolling team exploits an adversary's partial view of the perimeter to make the defending team appear stronger than it is. In~\cite{MOK-MO-UT:23}, an agent attempts to perform a hidden task while under partial observation by a supervisor. In~\cite{AS-CNH:14}, the authors propose three metrics for opacity which measures the ability for a partial observer to identify that a system has entered a set of secret states. The deception in this paper emerges as a structural consequence of lifting. Because the adversary observes only the graph and is unaware of the true MC state space, they can underestimate the patroller's performance. Therefore, deception incurs no additional computational cost beyond that of optimizing the lifted MC for performance. 

\subsection{Contributions}
This paper contains several contributions to the literature on stochastic patrolling. We propose the use of MC lifting to enable higher performance patrol strategies than can be achieved by the standard approach of assigning a single MC state to each graph node. We define new metrics based on the Kemeny constant, Stackelberg game capture probability, and return-time entropy that appropriately capture the performance of lifted MCs. We prove that optimized lifted MCs cannot perform worse than optimized non-lifted MCs and prove a performance improvement ceiling for lifted Kemeny constant optimization. This leads to a result of independent interest that relates the conductance and Kemeny constant for general, not necessarily reversible, MCs. We also show that MC lifting provides an interesting deceptive aspect since an adversary cannot perceive the state of the MC directly from the location of the patroller. This can lead to a model misspecification error on the part of the adversary. To enable the practical application of lifted MCs, we provide a tractable procedure for optimizing lifted MCs on travel-time-weighted graphs. Finally, we include results from several numerical experiments demonstrating the real-world performance improvement via the proposed procedure.

\section{Markov Chains}
We use a discrete-time MC to represent a stochastic patrol strategy. The patrol environment is encoded via a digraph $\mathcal{G} = (\mathcal{V},\mathcal{E})$ where nodes in $\mathcal{V}$ and edges in $\mathcal{E}$ represent locations to patrol and paths between them, respectively. We construct MCs such that (at least) one state corresponds to each node. We enforce that transition probabilities $p_{ij}$ between states $i,j$ are zero if the graph does not contain an edge between the corresponding nodes. For a state space $\Snl = \{1,\dots,n\}$, the MC can be represented as a row-stochastic transition matrix $P\in \real^{n\times n}$ where each entry $p_{ij}$ represents the probability that a patroller at state $i$ moves to state $j$ next. Let $X_k \in \Snl$ be the state of the MC $P$ at time step $k$. 

We restrict attention to irreducible MCs, i.e., those in which every state is reachable from every other state. It is well-known that an irreducible MC $P$ has a unique stationary distribution $\pi \in \real^{n}$ that satisfies $\pi^\top P = \pi^\top$ with $\pi^\top \1_n = 1$ and $\pi_i > 0$ for all $i$. The stationary distribution captures the limiting behavior of the MC~\cite[Thm. 2.1]{DA-JAF:02}:
\begin{equation}\label{eq:stat_dist}
    \lim_{T\to\infty} \frac{n_i(\mathbf{x}_T)}{T} = \pi_i, \quad \textup{a.s.,}
\end{equation}
where $n_i(\mathbf{x}_T)$ denotes the number of visits to state $i$ in a state trajectory $\mathbf{x}_T = (X_0,\dots,X_{T-1})$. Therefore, of all the patroller's long-term state visits, a fraction $\pi_i$ of them are to state $i$. If the nodes have varying priority, then that can be encoded via the stationary distribution. An alternative representation of an MC is the \emph{ergodic flow matrix} $Q = \Pi P$ where $\Pi = \textup{diag}(\pi)$. $Q$ satisfies $\1_n^\top Q \1_n = 1$ and $Q^\top \1_n = Q \1_n = \pi$. 

We introduce the \emph{first passage time} random variable $T_{i}^j = \min \{k \mid X_0 = i, X_k = j, k \ge 1\}$ that represents the number of time steps between the patroller leaving state $i$ and its first arrival at state $j$. We denote the \emph{mean first passage time} $m_{i}^j = \E[T_{i}^j]$. The \emph{first passage time probability matrices} $F_k$ contain the first passage time probabilities, i.e., $F_k(i,j) = \Pr(T_{i}^j = k)$. The $F_k$ matrices satisfy the recursion \cite[Ch.~5, Eq.~(2.4)]{EC:13}:
\begin{equation}\label{eq:F_k}
    F_{k+1} = P[F_k - \textup{diag}(F_k)],
\end{equation}
where $F_1 = P$. 

\subsection{Metrics for Markov Chain Patrol Strategies}\label{sec:metrics}
In this section we define three important MC patrol strategy metrics from the literature.
\subsubsection{Kemeny Constant}
The \emph{Kemeny constant} is a classical quantity first identified in~\cite{JGK-JLS:76} and proposed as a metric for MC patrol strategies in~\cite{RP-PA-FB:14b}. It is defined for an MC $P$ as 
    \begin{equation}\label{eq:kemeny_matrix}
    \kemeny(P) = \pi^\top M \pi,
\end{equation}
where $M \in \real^{n\times n}$ is the mean first passage time matrix whose entries are $m_{i}^j$. Intuitively, by minimizing the Kemeny constant of an MC, we design patrolling strategies that move quickly around the graph. Each column $m^j$ of $M$ can be computed as
\begin{equation}\label{eq:m_j}
    m^j = [I_n - P(I_n - E_j)]^{-1} \1_n,
\end{equation}
where $E_j = \textup{diag}(e_j)$ and $e_j$ is the $j$-th standard basis vector. 

\subsubsection{Stackelberg Game Capture Probability}
The \emph{Stackelberg Game Capture Probability} (SGCP) metric was proposed in~\cite{XD-DP-FB:19b} and extended in~\cite{YJ-GDG-XD-JRM-FB:23q}. This game-theoretic model identifies the optimal MC strategy for a patroller facing an omniscient adversary. In this model, the adversary has to remain undisturbed at a node $j$ for a given duration $\tau_j \in \natural_{>0}$ in order to complete an attack and win the game. If the patroller visits the attacked node $j$ within $\tau_j$ time periods, then the patroller wins. The metric (to be maximized) is
\begin{equation*}
    \mathcal{J}(P) = \min_{i,j \in \Snl} \left\{ \Pr \left[ T_{i}^j(P) \leq \tau_j \right] \right\}.
\end{equation*}
The minimization reflects the adversary's choice of a node $j$ to attack while the patroller is at node $i$. The probabilities of the patroller visiting node $j$ from node $i$ within $\tau_j$ time periods can be calculated from the first passage time probability matrices $F_k$ and collected into a capture probability matrix 
\begin{equation}
\Psi = \sum_{j=1}^n \sum_{k=1}^{\tau_j} F_k e_j e_j^\top. 
\end{equation}
The patroller's goal is to maximize the worst-case capture probability, i.e., the minimum entry of $\Psi$. 

\subsubsection{Return-Time Entropy}
The \emph{Return-Time Entropy} (RTE) metric was proposed in~\cite{XD-MG-FB:17o}. Consider an adversary who discreetly observes their target node for a period of time prior to attacking. If the patroller is absent for predictable intervals, then the adversary can act accordingly. To prevent this, the RTE metric identifies patrol strategies with unpredictable return times $T_j^j$ to each node $j$ in the graph. Unfortunately, the RTE does not admit a closed form, so we maximize the truncated RTE defined as
\begin{equation} \label{eq:ret_ent}
    \mathcal{H}(P) = -\sum_{j=1}^n \pi_j \sum_{k=1}^{K_\eta} F_k(j,j) \log F_k(j,j),
\end{equation}
where $F_k(j,j) = \Pr[T_j^j = k]$ and $K_{\eta} = \lceil 1 / (\eta \pi_{\textup{min}}) \rceil - 1$. $\eta \in (0,1)$ is a parameter that upper bounds the truncated probability, and $\pi_{\textup{min}}$ denotes the smallest entry of $\pi$.

\section{Lifted Markov Chains}\label{sec:lifting}
Lifting is a classical technique originally used to decrease the mixing time of an MC. Lifting an MC increases the size of its state space by creating virtual copies of the original states while preserving the stationary distribution of the original MC. Note that the literature contains other definitions of MC lifting, but in this paper we use the term to refer specifically to the following construction from~\cite{FC-LL-IP:99}.

To formally define lifting consider the MC $\Pl$ on the enlarged state space $\Sl = \{1,\dots,\nl\}$ and the MC $\Pnl$ on a partition $\Snl = \{S_1,\dots,S_\nnl\}$ of $\Sl$. The partition is specified via the indicator matrix $V \in \{0,1\}^{\nl \times \nnl}$ where the entry $v_{ij} = \bm{1}_{S_j}(i)$ indicates whether or not virtual state $i$ belongs to set $S_j$. Each row of $V$ sums to one and each column sum is at least one. We say that $\Pl$ is a \emph{lifting} of $\Pnl$ with respect to $V$ if the following conditions hold for all $i,j \in \Snl$~\cite{FC-LL-IP:99}:
\begin{equation}\label{eq:coll_elementwise}
\begin{aligned}
    \pinl_j & = \pil(S_j), \\
    \pnl_{ij} & = \frac{1}{\pil(S_i)} \sum_{k \in S_i} \pil_k \sum_{l \in S_j} \pl_{kl},
\end{aligned}
\end{equation}
where $\pil(S) = \sum_{i \in S} \pil_i$ for a set $S$. We call $\Pnl$ the \emph{collapsing} of $\Pl$ with respect to $V$. 
The first condition says that the stationary probability of each graph node is the sum of the stationary probabilities of the corresponding virtual copies. In this sense, lifting preserves the stationary distribution. The second condition computes collapsed transition probabilities by summing transition probabilities \textit{to} virtual copies and averaging (weighted by $\pil$) the transition probabilities \textit{from} virtual copies. Both conditions in Eq.~\eqref{eq:coll_elementwise} can be rewritten in matrix form by introducing $V$:
\begin{equation}\label{eq:coll}
\begin{aligned}
    \pinl & = V^\top \pil, \\
    \Pnl & = U \Pl V,
\end{aligned}
\end{equation}
where $U = \left( V^\top \Pi^\textup{lift} V \right)^{-1} V^\top \Pi^\textup{lift}$. The transition probability condition can also be written as a relationship between the ergodic flows $\Ql = \Pil\Pl, \Qnl = \Pinl\Pnl$:
\begin{equation}\label{eq:coll_Q}
    \Qnl = V^\top \Ql V.
\end{equation}
Notice that right-multiplying the above equation by $\1_\nnl$ yields the stationary distribution relationship in Eq.~\eqref{eq:coll}. Therefore, Eq.~\eqref{eq:coll_Q} captures both conditions in Eq.~\eqref{eq:coll}. 

\subsection{Corridor Environment: Mean Capture Time}
In this section we demonstrate the need for new metrics to appropriately capture the performance of lifted MCs. We return to the corridor environment represented by the four-node path graph and use for demonstration the non-lifted MC $\Pnl$ shown in Fig.~\ref{fig:motivating_example}(a) and the lifted MC $\Pl$ shown in Fig.~\ref{fig:motivating_example}(b) with $\pl = 0.9$. We first show that the Kemeny constant is a meaningful metric for non-lifted MCs. We perform repeated Monte Carlo runs as follows. We randomly select an initial node for the patroller and the adversary according to the stationary distribution $\pinl$. We simulate the movements of the patroller according to $\Pnl$ and count the number of transitions until the patroller arrives at the adversary's location. If both the patroller and adversary start at the same node, then we count the transitions until the patroller returns to that node. Fig.~\ref{fig:lifted_kemeny_constant} shows that the mean capture time computed this way converges to the Kemeny constant $\kemeny(\Pnl)$. Now we repeat the process using the lifted MC $\Pl$. Here we randomly select the initial virtual state of the patroller according to the stationary distribution $\pil$. Fig.~\ref{fig:lifted_kemeny_constant} shows that the mean capture time does not converge to the Kemeny constant $\kemeny(\Pl)$ but instead to another quantity $\liftK(\Pl)$ that we will define subsequently. 

\begin{figure}
\centering
\includegraphics[width=\linewidth]{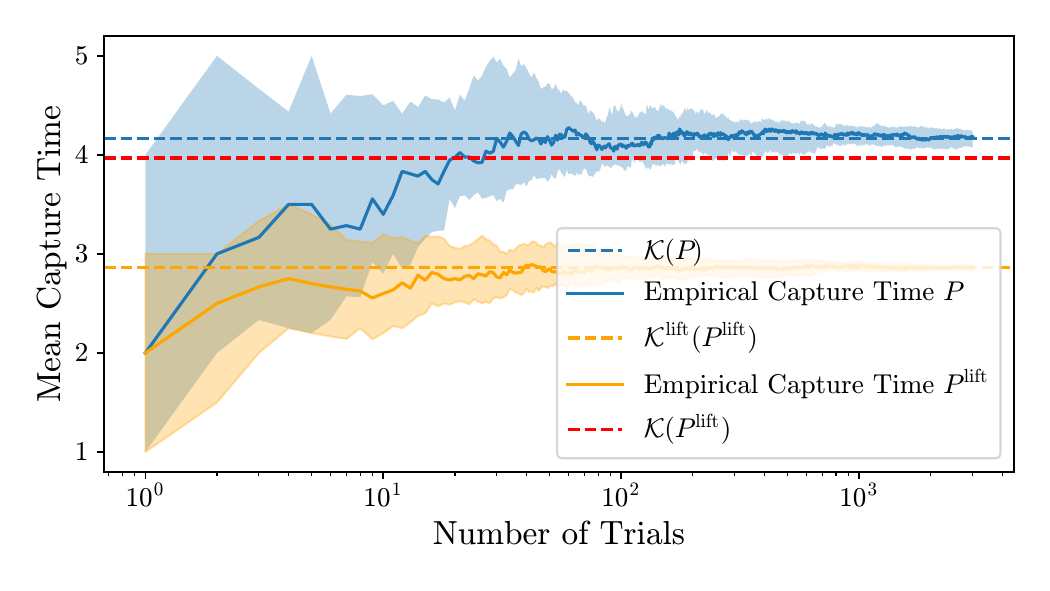}
\caption{Semi-log plot of the empirical mean capture time as a function of number of trials for the MC $\Pnl$ from Fig.~\ref{fig:motivating_example}(a) and the lifted MC $\Pl$ from Fig.~\ref{fig:motivating_example}(b) with $\pl=0.9$. The median and interquartile range are shown from 50 Monte Carlo runs. Also shown are the MC strategy metrics $\kemeny(\Pnl),\liftK(\Pl),\kemeny(\Pl)$. The empirical capture time under $\Pnl$ converges to $\kemeny(\Pnl)$ while the empirical capture time under $\Pl$ converges to $\liftK(\Pl)$ not $\kemeny(\Pl)$.}
\label{fig:lifted_kemeny_constant}
\end{figure}

\subsection{Metrics for Lifted Markov Chain Patrol Strategies}\label{sec:lifted_metrics}
To define appropriate metrics for lifted MCs we will need the set analogs of the first passage time quantities. The \emph{set first passage time} random variable $T_i^{S_j} = \min\{k \mid \Xl_0 = i, \Xl_k \in S_j, k \ge 1\}$ for a given set of states $S_j$, and the \emph{set mean first passage time} $m_{i}^{S_j} = \E[T_{i}^{S_j}]$. We will denote the \emph{set} first passage time probability matrices $F_k^\textup{lift}$ where $F_k^\textup{lift}(i,S_j) = \Pr( T_i^{S_j} = k )$. First-step analysis yields the following recursion for the matrices $F_k^\textup{lift}$ where $\circ$ denotes the Hadamard product:
\begin{equation}\label{eq:F_k_lift}
    F_{k+1}^\textup{lift} = \Pl \left[ F_k^\textup{lift} - \left( F_k^\textup{lift} \circ V \right) \right]
\end{equation}
with $F_1^\textup{lift} = \Pl V$. Now we can define the lifted MC metrics. The definitions will guarantee that the lifted metrics are equivalent to the original metrics for $V = I_\nnl$.

\subsubsection{Lifted Kemeny Constant}
The Kemeny constant, defined in Eq.~\eqref{eq:kemeny_matrix}, measures how quickly a patroller can cover a graph by averaging the mean first passage times $m_i^j$ between all pairs of nodes $i,j$ weighted by the stationary distribution. However, using this metric on lifted MCs results in a measure of the average passage time for a patroller to cover any pair of states in $\Sl$. This is not an appropriate metric since there are multiple states in $\Sl$ that represent the same node in the graph. Instead, the \textit{lifted Kemeny constant} will measure how fast a patroller can reach any of the nodes in $\Snl$ from any virtual state in $\Sl$. We define the lifted Kemeny constant as
\begin{equation}\label{eq:liftK}
    \liftK(\Pl) =  \left( \pil \right)^\top M^\textup{lift} V^\top \pil,
\end{equation}
where $M^\textup{lift} \in \real^{\nl\times\nnl}$ is the \emph{set} mean first passage time matrix with entries $m_i^{S_j}$. Each set $S_j$ contains all the virtual states corresponding to node $j$. Similarly to $M$, we compute $M^\textup{lift}$ column-wise:
\begin{equation}\label{eq:m_Sj}
    m^{S_j} = [I_{\nl} - \Pl(I_{\nl} - D_j)]^{-1} \1_{\nl},
\end{equation}
where $D_j = \textup{diag}(v^j)$ and $v^j$ denotes the $j$-th column of $V$. As shown in Fig.~\ref{fig:lifted_kemeny_constant}, the lifted Kemeny constant accurately computes the mean capture time of a lifted MC. Also note that because $m_i^{S_j} \leq m_i^l$ for all $l\in S_j$, we have that $\liftK(\Pl) \leq \kemeny(\Pl)$. An example of this can be seen in Fig.~\ref{fig:p_vs_M}. Combined with the findings in Fig.~\ref{fig:lifted_kemeny_constant}, we can conclude that the Kemeny constant overestimates the mean capture time of a lifted MC.

\subsubsection{Lifted Stackelberg Game Capture Probability}
Following the same reasoning, we propose the \textit{lifted SGCP} metric based on the set first passage times as follows: 
\begin{equation*}
    \mathcal{J}^\textup{lift}(\Pl) = \min_{i,j \in \Snl} \left\{ \frac{1}{\pinl_i} \sum_{l \in S_i} \pil_l \Pr\left[T_l^{S_j}(\Pl) \leq \tau_j\right] \right\},
\end{equation*}
where we use the $\pil$-weighted average of $T_l^{S_j}$ for each $l\in S_i$ because the adversary cannot distinguish the virtual states corresponding to each node, and the MC $\Pl$ is assumed to have reached stationarity. The corresponding lifted capture probability matrix is given by 
\begin{equation}\label{eq:Psi_lift}
    \Psi^\textup{lift} = U \sum_{j=1}^\nnl \sum_{k=1}^{\tau_j} F_k^\textup{lift} e_j e_j^\top,
\end{equation}
where $e_j \in \real^\nnl$ and $\Psi^\textup{lift} \in \real^{\nnl\times\nnl}$. The patroller aims to maximize the minimum entry of $\Psi^\textup{lift}$. The conceptual difference between this metric and the SGCP is that here we are departing from the setting of an omniscient adversary. Here we assume that the adversary only knows the location of the patroller on the graph, not its virtual state.

\subsubsection{Lifted Return-Time Entropy}\label{sec:lift_RTE}
We propose the \textit{lifted RTE} metric that measures the unpredictability of return times to each node $j$ via the entropy of the set return times $T_l^{S_j}$ for $l\in S_j$. At stationarity and with an adversary that cannot distinguish the virtual copies, the probability of returning to node $j$ at time step $k$ is
\begin{equation}\label{eq:R_k}
    R_k(S_j) =  \frac{1}{\pinl_j} \sum_{l\in S_j} \pil_l \Pr[T_l^{S_j} = k],
\end{equation}
where $\Pr[T_l^{S_j} = k] = F_k^{\mathrm{lift}}(l,S_j)$. Then we can define the truncated lifted RTE as
\begin{equation}\label{eq:H_lift}
    \mathcal H^{\mathrm{lift}}(\Pl)
    =-\sum_{j=1}^\nnl \pinl_j\sum_{k=1}^{K_\eta}R_k(S_j) \log R_k(S_j).
\end{equation}
For a meaningful comparison between the RTE and the lifted RTE, we must ensure that the expected return times are equal. For a non-lifted MC, Kac's lemma~\cite[Lem.~21.12]{DaL-YP:17} directly yields the expected return time $m_j^j = 1/\pinl_j$ for all $j$. For a lifted MC, we have that the expected lifted return time
\[
    \sum_{k = 1}^\infty k R_k(S_j) = \frac{1}{\pinl_j} \sum_{l\in S_j} \pil_l m_l^{S_j} = \frac{1}{\pinl_j},
\]
where the second equality is due to another application of Kac's lemma. Therefore, the expected return times are the same, and it is reasonable to compare their entropies. This also yields an expression for $K_{\eta} = \lceil 1 / (\eta \pinl_{\textup{min}}) \rceil - 1$, guaranteeing that $\eta$ upper bounds the truncated probability by a similar argument as in~\cite{XD-MG-FB:17o}. 

\subsection{Corridor Environment: Lifted Kemeny Constant}
Once again, we return to the corridor environment and use the lifted Kemeny constant to more accurately capture the performance of $\Pl$ from Fig.~\ref{fig:motivating_example}(b). Fig.~\ref{fig:p_vs_M} shows $\liftK(\Pl)$ as a function of $\pl$. It can be seen that $\Pl$ actually performs better than $\Pnl$ from Fig.~\ref{fig:motivating_example}(a) for all values of $\pl \neq 1/2$ and equally well at $\pl = 1/2$. It is desirable that $\liftK(\Pl) = \kemeny(\Pnl)$ for $\pl = 1/2$ because the behavior of the two chains on the graph is indistinguishable. The highest-performing lifted strategies at $\pl \in \{0,1\}$ achieve $\liftK(\Pl) = 47/18$, a dramatic improvement over $\kemeny(\Pnl) = 25/6$. It can also be shown that the lifted Kemeny constant of MCs with the same structure as $\Pl$ and $\pl \in \{0,1\}$ scales linearly with $\nnl$ and with a smaller constant factor than $\kemeny(\Pl)$. These findings are summarized in the third column of Table~\ref{tab:running_example}.

\section{Quantifying Improvement from Lifting}
The previous sections highlighted the advantages of using lifted MCs as stochastic patrol strategies and defined appropriate metrics for measuring their performance. In this section we will quantify the performance improvement possible via lifting. To do so, we will need the concept of MC conductance. The \emph{conductance} $\Phi_A(P)$ of a set of states $A$ for a given MC $P$ is defined as $\Phi_A(P) = Q(A,A^c)/[\pi(A)\pi(A^c)]$ where $A^c$ denotes the complement of $A$ and $Q(A,A^c) = \sum_{i\in A} \sum_{j\in A^c} q_{ij}$~\cite{FC-LL-IP:99}. The conductance of the MC $\Phi(P) = \min_{A \subsetneq \Snl,A \neq \emptyset} \Phi_A(P)$. $\Phi$ is known to possess a monotonicity property with respect to lifting, in particular, $\Phi(\Pl) \leq \Phi(\Pnl)$ for any choice of $V$~\cite{FC-LL-IP:99}. $\Phi$ is a classical flow parameter that can be interpreted as a measure of the ease of movement around the state space. Large values indicate easy travel throughout the state space while small values indicate the presence of a bottleneck. For non-lifted MCs $\Pnl$ that represent stochastic patrol strategies on a graph, $\Phi(\Pnl)$ is a function of both the graph topology and the transition probabilities on the available edges. We will be able to show that the ease of movement under $\Pnl$, as measured by $\Phi(\Pnl)$, constrains the achievable performance of liftings of $\Pnl$.

\subsection{Main Results}
The following theorem establishes bounds on the best lifted Kemeny constant achievable for a lifting of a given MC $\Pnl$. The upper bound provides a guarantee of performance not degrading as a result of lifting. The lower bound is valuable because it enables the estimation of possible performance improvement via lifting. The lower bound is given in terms of the MC conductance $\Phi(\Pnl)$ which measures bottlenecks of $\Pnl$.
\begin{theorem}[Lifted Kemeny Constant Bounds]\label{thm:bounds}
    For any irreducible MC $\Pnl$ and partition $V$, the optimal lifting achieves
    \begin{equation}\label{eq:perf_bounds}
        \frac{1}{2 \Phi(\Pnl)} \leq \min_{\Pl \in \mathcal{L}(\Pnl,V)} \liftK(\Pl) \leq \kemeny(\Pnl),
    \end{equation}
    where $\mathcal{L}(\Pnl,V)$ denotes the set of all irreducible liftings of $\Pnl$ with respect to $V$.
\end{theorem}

The proofs of the upper and lower bounds will be presented separately in Sections~\ref{sec:upper_bound} and \ref{sec:lower_bound}. Note that the bounds are functions of the given MC $\Pnl$ so the possible performance improvement is determined by $\Pnl$. Furthermore, this result holds for any non-lifted MC including one that is optimized for a particular graph. We will be able to adapt the upper bound argument from Theorem~\ref{thm:bounds} to prove the following corollary for the remaining metrics. Recall that the goal is to maximize these metrics. The proof is given in Section~\ref{sec:upper_bound}.

\begin{corollary}[Lifted Metric Improvement]\label{cor:lmi}
    The optimal lifting of a given MC $\Pnl$ with respect to partition $V$ achieves 
    \begin{equation}
    \begin{aligned}
        \mathcal{J}(\Pnl) & \leq \max_{\Pl \in \mathcal{L}(\Pnl,V)} \mathcal{J}^\textup{lift}(\Pl), \\
        \mathcal{H}(\Pnl) & \leq \max_{\Pl \in \mathcal{L}(\Pnl,V)} \mathcal{H}^\textup{lift}(\Pl),
    \end{aligned}
    \end{equation} 
    where $\mathcal{L}(\Pnl,V)$ denotes the set of all irreducible liftings of $\Pnl$ with respect to $V$.
\end{corollary}

\subsection{Corridor Environment: Conductance Lower Bound}
Returning to our running example of the corridor environment, we can compute the conductance bound. For the non-lifted MC $\Pnl$ shown in Fig.~\ref{fig:motivating_example}(a), the conductance bound $[2\Phi(\Pnl)]^{-1} = 3/4$. The gap between the Kemeny constant $\kemeny(\Pnl) = 25/6$ and the conductance bound indicates the opportunity for performance improvement via lifting. For general path graphs with corresponding MCs $\Pnl$ of the form of Fig.~\ref{fig:motivating_example}(a), it can be shown that the conductance bound scales linearly with $\nnl$. Therefore, we know that linear scaling is the best possible via lifting. $\Pl$ shown in Fig.~\ref{fig:motivating_example}(b) with $\pl \in \{0,1\}$ achieves linear scaling. The last two columns of Table~\ref{tab:running_example} summarize these findings.

\subsection{Upper Bound in Theorem~\ref{thm:bounds}}\label{sec:upper_bound}
In this section we will prove the upper bound presented in Theorem~\ref{thm:bounds} using the concept of lumpability defined in Appendix~\ref{app:lumpability}. We begin with the following lemma whose proof is given in Appendix~\ref{app:lem_3}.
\begin{lemma}\label{lem:lump}
    If an irreducible MC $\Pl$ is lumpable with respect to $V$, then 
    \begin{align}
        F_k^{\mathrm{lift}} &= V F_k, \quad \forall k \in \mathbb{N}_{>0}, \label{eq:VF} \\
        M^{\mathrm{lift}} &= V M, \label{eq:VM}
    \end{align}
    where $F_k$ and $M$ are the first passage time probability matrices and mean first passage time matrix of the collapsed MC $\Pnl$, respectively.
\end{lemma}
An important consequence of Lemma~\ref{lem:lump} is the relationship between $\liftK(\Pl)$ and $\kemeny(\Pnl)$ for lumpable MCs $\Pl$: 
\begin{align}\label{eq:2kemeny}
    \liftK(\Pl) & = \left( \pil \right)^\top M^\textup{lift} V^\top \pil = \left( \pil \right)^\top V M V^\top \pil \notag \\
    & = \pinl^\top M \pinl = \kemeny( \Pnl ).
\end{align}
This ensures that $\liftK(\Pl) = \kemeny(\Pnl)$ if $\Pl$ is lumpable with respect to $V$. This is a desirable property because $\Pnl$ exactly reproduces the behavior of $\Pl$. We can use this property to prove the upper bound. 
\begin{proof}[Proof of Upper Bound in Theorem~\ref{thm:bounds}]
    Let us define the optimization problem in Eq.~\eqref{eq:perf_bounds} that finds the optimal lifting of a given non-lifted MC $\Pnl$ with partition $V$,
    \begin{equation} \label{eq:opt:P} \tag{P1}
    \begin{aligned}
        \mathcal{F}_1(\Pnl) = \min_{\Pl,\pil} \quad & \liftK(\Pl) \\
        \text{s.t.} \quad & \Pl \text{ is an irreducible MC}, \\
        & \left( \pil \right)^\top \Pl = \left( \pil \right)^\top, \\ 
        & \pinl = V^\top \pil, \\
        & \Pnl = U \Pl V,
    \end{aligned}
    \end{equation}
    and a second version that enforces lumpability\footnote{Note that \eqref{eq:opt:lump} is feasible because we can construct a feasible $\Pl$ by copying rows of $\Pnl$ corresponding to virtual copies and dividing columns evenly among virtual copies.}, 
    \begin{equation} \label{eq:opt:lump} \tag{P2}
    \begin{aligned}
        \mathcal{F}_2(\Pnl) = \min_{\Pl,\pil} \quad & \liftK(\Pl) \\
        \text{s.t.} \quad & \Pl \text{ is an irreducible MC}, \\
        & \left( \pil \right)^\top \Pl = \left( \pil \right)^\top, \\ 
        & \pinl = V^\top \pil, \\
        & V \Pnl = \Pl V, \\ 
        & \Pnl = U \Pl V.
    \end{aligned}
    \end{equation}
    We also define a non-lifted MC optimization problem whose optimal value is trivially $\kemeny(\Pnl)$,
    \begin{equation} \label{eq:opt:barP} \tag{P3}
    \begin{aligned}
        \mathcal{F}_3(\Pnl) = \min_{\tilde{P},\tilde{\pi}} \quad & \kemeny(\tilde{P}) \\
        \text{s.t.} \quad & \tilde{P} \text{ is an irreducible MC}, \\ 
        & \tilde{\pi}^\top \tilde{P} = \tilde{\pi}^\top, \\ 
        & \pinl = \tilde{\pi}, \\
        & \Pnl = \tilde{P}.
    \end{aligned}
    \end{equation}

    Clearly, the only difference between problems~\eqref{eq:opt:P} and \eqref{eq:opt:lump} is the set of linear constraints $V \Pnl = \Pl V$ that  enforce lumpability of $\Pl$ for the given partition $V$. Therefore, the optimal value $\mathcal{F}_1(\Pnl)$ attained by problem \eqref{eq:opt:P} is not greater than the optimal value $\mathcal{F}_2(\Pnl)$ attained by problem \eqref{eq:opt:lump} because the feasible region of \eqref{eq:opt:lump} is contained in that of \eqref{eq:opt:P}. Also, Lemma~\ref{lem:lump} and Eq.~\eqref{eq:2kemeny} imply that $\mathcal{F}_2(\Pnl) = \mathcal{F}_3(\Pnl)$. Then, by transitivity, we have that $\mathcal{F}_1(\Pnl) \leq \mathcal{F}_3(\Pnl) = \kemeny(\Pnl)$. 
\end{proof}

The benefit of constructing the argument this way is that we can also prove the more general statement that $\min_{\Pl} \liftK(\Pl) \leq \min_{ \tilde{P} } \kemeny( \tilde{P} )$ for irreducible MCs meeting the corresponding stationary distribution constraints. Simply remove the last constraint in problems \eqref{eq:opt:P}--\eqref{eq:opt:barP} and set $\Pnl$ in \eqref{eq:opt:lump} equal to the minimizer of \eqref{eq:opt:barP}, and the same argument applies. We can also make the same argument for the lifted SGCP and lifted RTE metrics. We just need to show, for both metrics, that a lumpable MC $\Pl$ achieves the same value as the lumped MC $\Pnl$. 

\begin{proof}[Proof of Corollary~\ref{cor:lmi}]
    If we can show that lumpability of $\Pl$ with respect to $V$ implies that $\mathcal{J}^\textup{lift}(\Pl) = \mathcal{J}(\Pnl)$ and $\mathcal{H}^\textup{lift}(\Pl) = \mathcal{H}(\Pnl)$, then the rest of the argument is identical to that of the Kemeny constant. Starting with the SGCP metric, we can use Lemma~\ref{lem:lump} to substitute $VF_k$ for $F_k^\textup{lift}$ in Eq.~\eqref{eq:Psi_lift}. Then using $UV = I_\nnl$ yields $\Psi^\textup{lift} = \Psi$ which implies that $\mathcal{J}^\textup{lift}(\Pl) = \mathcal{J}(\Pnl)$. For the RTE metric, using Lemma~\ref{lem:lump} on Eq.~\eqref{eq:R_k} implies that $R_k(S_j) = F_k(j,j)$ which yields $\mathcal{H}^\textup{lift}(\Pl) = \mathcal{H}(\Pnl)$. 
\end{proof}

\subsection{Lower Bound in Theorem~\ref{thm:bounds}}\label{sec:lower_bound}
In this section we will prove the lower bound presented in Theorem~\ref{thm:bounds}. First, we will need the following lemma whose proof is given in Appendix~\ref{app:lem_4}.
\begin{lemma}\label{lem:m_Q_rel}
    For any irreducible MC $P$, the stationary-distribution-weighted sum of mean first passage times from a non-empty set $A \subsetneq \Snl$ to its complement $A^c$ satisfies 
    \begin{equation}\label{eq:m_Q_rel}
    \sum_{i \in A} \pi_i m_i^{A^c} \geq \frac{\pi(A)^2}{2Q(A,A^c)}.
    \end{equation}
\end{lemma}

Now we can prove the lower bound in Theorem~\ref{thm:bounds}.
\begin{proof}[Proof of Lower Bound in Theorem~\ref{thm:bounds}]
    First, we will show that the ergodic flow has a symmetry property, i.e., $Q(A,A^c) = Q(A^c,A)$ where $A$ is a non-empty proper subset of the nodes and $A^c$ is its complement. From the property $Q^\top \1_n = Q \1_n = \pi$, we have that
    \[
    \begin{aligned}
    \pi(A) & = \sum_{i \in A} \sum_{j \in A} q_{ij} + \sum_{i \in A} \sum_{j \in A^c} q_{ij} = Q(A,A) + Q(A,A^c), \\
    \pi(A) & = \sum_{i \in A} \sum_{j \in A} q_{ij} + \sum_{i \in A^c} \sum_{j \in A} q_{ij} = Q(A,A) + Q(A^c,A),
    \end{aligned}
    \]
    which implies that $Q(A,A^c) = Q(A^c,A)$. 

    Now we are ready to lower bound $\liftK$. We suppress the dependence on $\Pl$ for brevity. From $V$, construct $S_A = \cup_{j \in A} S_j$, the set of all virtual states in $\Sl$ corresponding to the nodes in $A$. From the definition, we have
    \[
    \begin{aligned}
    \liftK & = \sum_{i \in \Sl} \pil_i \sum_{j \in \Snl} m_i^{S_j} \pinl_j \\
    & \geq \sum_{i \in S_A} \pil_i \sum_{j \in A^c} m_i^{S_j} \pinl_j + \sum_{i \in S_{A^c}} \pil_i \sum_{j \in A} m_i^{S_j} \pinl_j.
    \end{aligned}
    \]
    Note that $m_i^{S_j} \geq m_i^{S_{A^c}}$ for all $j \in A^c$ and $m_i^{S_j} \geq m_i^{S_A}$ for all $j \in A$. Therefore we can lower bound each term as
    \[
    \begin{aligned}
    \sum_{i \in S_A} \pil_i \sum_{j \in A^c} m_i^{S_j} \pinl_j & \geq \pinl(A^c) \sum_{i \in S_A} \pil_i m_i^{S_{A^c}}, \\
    \sum_{i \in S_{A^c}} \pil_i \sum_{j \in A} m_i^{S_j} \pinl_j & \geq \pinl(A) \sum_{i \in S_{A^c}} \pil_i m_i^{S_A}.
    \end{aligned}
    \]
    Then using Lemma~\ref{lem:m_Q_rel}, the symmetry of $Q$, and the definition of lifting we arrive at
    \[
    \begin{aligned}
    \liftK & \geq \frac{\pinl(A^c)\pil(S_A)^2}{2 \Ql(S_A,S_{A^c})} + \frac{\pinl(A)\pil(S_{A^c})^2}{2 \Ql(S_{A^c},S_A)} \\
    & = \frac{\pinl(A)\pinl(A^c)}{2 \Qnl(A,A^c)} = \frac{1}{2\Phi_A(\Pnl)}.
    \end{aligned}
    \]
    Optimizing over $A$ subject to $A \subsetneq \Snl, A\neq\emptyset$ yields the lower bound of $[2\Phi(\Pnl)]^{-1}$. Because this holds for any lifting of $\Pnl$ with respect to $V$, it also holds for the minimizing lifting.
\end{proof}
Theorem~\ref{thm:bounds} trivially yields the following corollary, but it can also be shown directly in a similar manner to the above proof. 
\begin{corollary}[Conductance--Kemeny Constant Relationship]\label{cor:cond_bound_K}
    For any irreducible MC $P$, the Kemeny constant $\kemeny$ is related to the conductance $\Phi$ as $[2 \Phi(P)]^{-1} \leq \kemeny(P)$.
\end{corollary}
This result is of independent interest because it provides an analog of the known relationship $\Phi(P_\textup{rev})^{-1}+1 \leq \kemeny(P_\textup{rev})$ for reversible MCs~\cite[Ch.~4]{DA-JAF:02}. Due to the monotonicity property of the conductance with respect to lifting, Corollary~\ref{cor:cond_bound_K} also yields the following bound on the possible improvement in the classical Kemeny constant.
\begin{corollary}[Kemeny Constant Bound]
    The Kemeny constant of the optimal lifted MC for a given collapsed MC $\Pnl$ and partition $V$ is related to the conductance $\Phi$ as $[2\Phi(\Pnl)]^{-1} \leq \min_{\Pl \in \mathcal{L}(\Pnl,V)} \kemeny(\Pl)$.
\end{corollary}

\section{Deception via Lifting}
In this section, we describe an additional benefit of constructing lifted MCs according to Eq.~\eqref{eq:coll}. Consider an adversary that assumes the patroller is using a standard non-lifted MC. If the patroller is actually using a lifted MC, then the adversary commits a \emph{model misspecification error}~\cite{LL:78b, HW:82}. We will refer to an adversary that makes this error as a \emph{naive} adversary. If the adversary attempts to estimate a non-lifted MC from the patroller's trajectory on the graph, then we can prove that standard estimators will converge to $\Pnl$. $\Pnl$ is also the best estimate in the following senses: it is the KL-optimal aggregation of $\Pl$ with respect to $V$~\cite{KD-PGM-SPM:11}, and it minimizes the KL divergence rate to the lumped chain of $\Pl$ (defined in Appendix~\ref{app:lumpability})~\cite{BCG-TP-GK-HK:15}. Now we prove that standard estimators yield $\Pnl$. We consider the following class of estimators.
\begin{defn}[Plug-In Estimator]\label{def:plugin}
    An MC estimator is a \emph{plug-in estimator} if it satisfies, for $i,j \in \Snl$, 
    \[
    \lim_{T\to\infty} \hat{p}_{ij}(\mathbf{x}_T) = \lim_{T\to\infty} \frac{n_{ij}(\mathbf{x}_T)}{n_i(\mathbf{x}_T)},
    \]
    where $\hat{p}_{ij}(\mathbf{x}_T)$ is the estimated transition probability based on a state trajectory $\mathbf{x}_T$ of length $T$ and $n_{ij},n_i$ are empirical transition counts and empirical visit counts, respectively.
\end{defn}
Examples of plug-in estimators include maximum likelihood estimators, Bayesian estimators or method of moments estimators, etc. Now we state the result; the proof is in Appendix~\ref{app:deception}. 
\begin{prop}[Plug-In Estimate Converges to Collapsing]\label{prop:deception}
    Let $\xnl_T$ be the projection onto $\Snl$ of a length-$T$ state trajectory $\xl_T$ on $\Sl$ of an irreducible lifting $\Pl$ of $\Pnl$ with respect to $V$. Then every plug-in estimate $\hat{P}$ satisfies
    \begin{equation}\label{eq:deception}
        \lim_{T\to\infty} \hat{P}(\mathbf{x}_T) = \Pnl, \quad \textup{a.s.}
    \end{equation}
\end{prop}

The takeaway is that a naive adversary that uses a plug-in estimator will arrive at $\Pnl$. By multiple definitions this is also the optimal estimate~\cite{KD-PGM-SPM:11, BCG-TP-GK-HK:15}. Then relying on the estimate $\Pnl$ when the patroller is actually using $\Pl$ can lead the adversary to underestimate the performance of the patroller. Even for a non-naive adversary, estimating $\Pl$ from its trajectory on $\Snl$ is difficult~\cite{YO:21}.

\subsection{Corridor Environment: Plug-In Estimate Convergence}
Returning to our running example, we can simulate $\Pl$ shown in Fig.~\ref{fig:motivating_example}(b) with, say, $\pl = 0.9$ and generate a virtual state trajectory $\xl_T$ and a projected graph trajectory $\xnl_T$. We then assume the role of the adversary that has observed the patroller and recorded $\xnl_T$. If the adversary naively performs maximum likelihood estimation of a non-lifted MC, the resulting estimate $\hat{P}$ converges to $\Pnl$ shown in Fig.~\ref{fig:motivating_example}(a) in the limit as the length of the recorded trajectory approaches infinity. Fig.~\ref{fig:T_vs_est_err} shows a Monte Carlo simulation of the convergence of the estimate $\hat{P}$ to $\Pnl$ as the trajectory length $T$ increases. If the adversary plans an attack based on $\hat{P}$, then they are liable to experience worse performance than expected. Because $\liftK(\Pl) < \kemeny(\Pnl)$ as shown in Fig.~\ref{fig:p_vs_M}, the adversary will have less time than they expect to execute an attack.

\begin{figure}
\centering
\includegraphics[width=\linewidth]{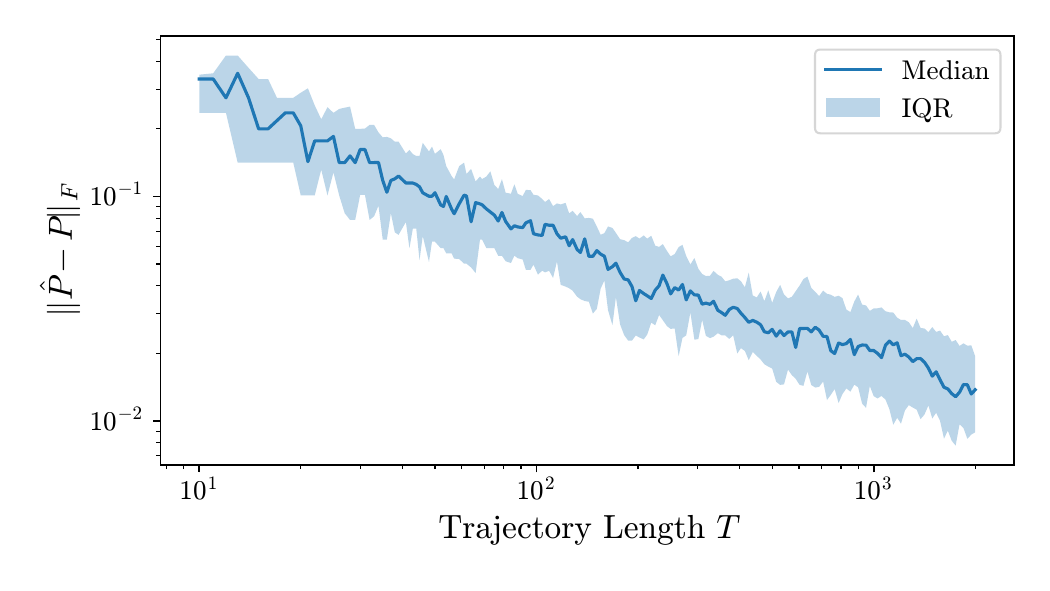}
\caption{Log-log plot of the Frobenius norm of the difference between the maximum likelihood estimate of the MC $\hat{P}$ and the collapsed MC $\Pnl$ from Fig.~\ref{fig:motivating_example}(a) as a function of trajectory length $T$. The median and interquartile range are shown from 50 Monte Carlo runs. It can be seen that $\hat{P}$ converges to $\Pnl$ as $T\to\infty$.}
\label{fig:T_vs_est_err}
\end{figure}

\section{Optimizing Lifted Markov Chains}\label{sec:lift_opt}
In this section, we present a method for optimizing lifted MCs. The method is applicable to all three lifted metrics. The problem data are the non-lifted MC $\Pnl$ and the indicator matrix $V$. The goal is to identify a lifting $\Pl$ of $\Pnl$ with respect to $V$ that optimizes the metric of interest $\mathcal{Y}$. This can be written as a minimization without loss of generality:
\begin{equation}\label{eq:opt:liftP}
    \min_{\Pl\in\mathcal{L}(\Pnl,V)} \mathcal{Y}(\Pl).
\end{equation}

Recall the nonlinear constraints ensuring a valid lifting given in Eq.~\eqref{eq:coll}. If we instead formulate the problem in terms of the ergodic flows $\Qnl,\Ql$, then the constraints become linear as shown in Eq.~\eqref{eq:coll_Q}. Therefore, we can rewrite the optimization problem~\eqref{eq:opt:liftP} as
\begin{equation} \label{eq:opt:liftQ} \begin{array}{cl} 
        \underset{ \Ql }{ \min } & \mathcal{Y}(\Ql)  \\
        \text{s.t.} & \Ql \1_{\nl} = \left( \Ql \right)^\top \1_{\nl}, \\
            & \epsilon V \lceil \Qnl \rceil V^\top \leq \Ql \leq V \lceil \Qnl \rceil V^\top, \\
            & V^\top \Ql V = \Qnl, \\
    \end{array} \end{equation}
where $0 < \epsilon \ll 1$ and $\lceil \Qnl \rceil$ denotes the element-wise ceiling of $\Qnl$. The first constraint (combined with the third) ensures a valid ergodic flow matrix $\Ql$. The lower bound in the second constraint ensures irreducibility, and the upper bound ensures graph-compatibility. The third constraint is Eq.~\eqref{eq:coll_Q} that ensures a valid lifting of $\Qnl$ with respect to $V$. We propose a Projected Gradient Descent (PGD) approach to optimize lifted MC metrics subject to the linear constraints given in Eq.~\eqref{eq:opt:liftQ}. In each iteration, the gradient descent step $\Qlt_{t+1} = \Ql_t - \alpha \frac{\partial \mathcal{Y}}{\partial \Ql}(\Ql_t)$ with step size $\alpha$ is followed by a projection step
\begin{equation} \label{eq:qp} \begin{array}{cl} 
        \underset{ \Ql_{t+1} }{ \min } & \lVert \Ql_{t+1} - \Qlt_{t+1} \rVert_F^2  \\
        \text{s.t.} 
            & \Ql_{t+1} \1_{\nl} = \left( \Ql_{t+1} \right)^\top \1_{\nl}, \\
            & \epsilon V \lceil \Qnl \rceil V^\top \leq \Ql_{t+1} \leq V \lceil \Qnl \rceil V^\top, \\
            & V^\top \Ql_{t+1} V = \Qnl, \\
    \end{array} \end{equation}
which is a convex quadratic program. 

\subsection{Travel Times}
Now we generalize the lifted metrics to accommodate weighted graphs where the weights represent edge travel times. We will reuse the same symbols because unweighted graphs can be handled by setting all weights to one. Consider the lifted weight matrix $\Wl = V \Wnl V^\top$ where the entries $\wnl_{ij} \in \natural_{>0}$ of $\Wnl$ represent the number of time steps it takes the patroller to traverse the edge from node $i$ to node $j$. We define the \emph{weighted} set first passage time
\[
    T_i^{S_j} = \min_k \left\{ \sum_{t=0}^{k-1} \wl_{\Xl_t,\Xl_{t+1}}\, \middle| \, \Xl_0 = i, \Xl_k \in S_j, k\geq1 \right\}. 
\]

For the lifted Kemeny constant, it can be shown that the \emph{weighted} set mean first passage times satisfy
\begin{equation}\label{eq:liftM}
    m^{S_j} = [I_{\nl} - \Pl(I_{\nl} - D_j)]^{-1} (\Pl \circ \Wl) \1_{\nl},
\end{equation}
where $m^{S_j}$ contains $m^{S_j}_i = \E [T_i^{S_j}]$ for all virtual states $i\in\Sl$. Then using $M^\textup{lift}$, the matrix whose entries are $m_i^{S_j}$, in Eq.~\eqref{eq:liftK} extends the lifted Kemeny constant to weighted graphs. Note that setting $\Wl = \1_{\nl\times\nl}$ in Eq.~\eqref{eq:liftM} recovers Eq.~\eqref{eq:m_Sj}. For the lifted SGCP metric and the lifted RTE metric, we can generalize the lifted first passage time probability recursion given in Eq.~\eqref{eq:F_k_lift}:
\begin{equation}
\begin{aligned}
    F_k^\textup{lift} & = [\Pl \circ \mathbf{1}_{k}(\Wl)] V \\
    & + \sum_{w=1}^{w_\textup{max}} [\Pl \circ \mathbf{1}_{w}(\Wl)] \left[F_{k-w}^\textup{lift} - \left( F_{k-w}^\textup{lift} \circ V \right) \right]
\end{aligned}
\end{equation}
where we define $F_k^\textup{lift} = \0_{\nl\times\nnl}$ for $k\leq0$ and $\mathbf{1}_{z}(\Wl)$ is an $\nl \times \nl$ matrix that takes the value one at entries of $\Wl$ equal to $z$ and zero otherwise. Calculating $F_k^\textup{lift}$ as above allows us to use Eq.~\eqref{eq:Psi_lift} and Eq.~\eqref{eq:H_lift} unchanged for computing the lifted SGCP and lifted RTE metrics with weighted graphs. 

For the lifted RTE metric thus defined, we need to ensure that the expected weighted lifted return time equals the expected weighted return time. Given the non-lifted MC $\Pnl$, it is known that $\E[T_j^j] = [\pinl^\top (\Pnl \circ \Wnl) \1_\nnl]/\pinl_j$ for all nodes $j$~\cite[Thm. 6]{RP-PA-FB:14b}. The following lemma whose proof is in Appendix~\ref{app:return_time} provides the desired result for lifted MCs.
\begin{lemma}\label{lem:mean_return_time}
    Given edge weights $\Wnl$, partition $V$, and an irreducible lifting $\Pl$ of the MC $\Pnl$, the expected lifted weighted return time satisfies $\sum_{k=1}^\infty k R_k(S_j) = [\pinl^\top (\Pnl \circ \Wnl) \1_\nnl]/\pinl_j$ for all $j \in \Snl$ where $R_k(S_j) = (1/\pinl_j) \sum_{l \in S_j} \pil_l \Pr[T_l^{S_j} = k]$ is a function of the weighted set first passage times. 
\end{lemma}
Therefore, the lifted RTE metric can be directly compared to the RTE metric on weighted graphs. We can also use Lemma~\ref{lem:mean_return_time} to choose $K_\eta = \lceil \wnl_\textup{max} / (\eta \pinl_\textup{min}) \rceil - 1$ and ensure that $\eta$ upper bounds the truncated probability by a similar argument as in~\cite{XD-MG-FB:17o}. 

\subsection{Gradient of the Lifted Kemeny Constant}
For the lifted Kemeny constant we are able to provide a closed form for the gradient. We will present the gradient of the lifted Kemeny constant for weighted graphs with the understanding that setting $\Wl = \1_{\nl\times\nl}$ will recover the gradient for unweighted graphs. The proof is given in Appendix~\ref{app:gradient}.
\begin{prop}[Gradient of the Lifted Kemeny Constant]\label{prop:grad_lift_K}
    The gradient of the lifted Kemeny constant for weighted graphs with respect to the ergodic flow $\Ql$ for fixed $\pinl$ is given by
    \begin{align}\label{eq:grad_liftK}
        \frac{\partial \liftK}{\partial \Ql} & = \sum_{j=1}^\nnl \Bigl[ \left( \pinl_j I_{\nl} + \Lambda_j \right) m^{S_j} \1_{\nl}^\top - \left(\lambda_j \1_{\nl}^\top\right) \circ \Wl \notag \\ & 
        - \lambda_j \left(m^{S_j}\right)^\top (I_{\nl} - D_j)   \Bigr],
    \end{align}
    where $\Lambda_j = \textup{diag}(\lambda_j)$ and
    \begin{equation}\label{eq:lambda}
        \lambda_j = -\pinl_j \left[ \Pil - (I_{\nl} - D_j) \left( \Ql \right)^\top \right]^{-1} \pil.
    \end{equation}
\end{prop}

\section{Numerical Results}
In this section we augment the provided theoretical results with numerical experiments demonstrating the effectiveness of lifted MCs as stochastic patrol strategies. Experiments were run on a machine with an Intel Core i7-10710U CPU (6 cores, 1.10 GHz) and 16 GB of RAM. The code is available at: 
\url{https://github.com/yohanjohn307/markov-chain-lifting}.

\subsection{Choosing the Partition}
Thus far we have considered the partition $V$ a free parameter chosen prior to optimization of the lifted MC. This is a subject of study in its own right~\cite{KJ-DS-JS:07}. Here we explore various approaches numerically. We focus on graphs with 10 nodes because large graphs lead to poor coverage for a single patroller. Our testbed consists of 50 Erdős--Rényi (ER) random undirected graphs (with self-loops and conditioned on connectivity) with randomly chosen edge probability $p \sim \textup{Uniform}(0.2, 0.8)$. For each graph, we also randomly sample a stationary distribution $\pinl \sim \textup{Dirichlet}(5\cdot\1_\nnl)$ that avoids extremely small probabilities at any node. For each random graph and stationary distribution, we minimize the Kemeny constant of a non-lifted MC $\Qnl$ via PGD\footnote{It can be shown that the gradient of the Kemeny constant for weighted graphs with respect to $\Qnl$ for fixed $\pinl$ is given by
\[
    \begin{aligned}
    \frac{\partial \kemeny}{\partial \Qnl} & = - \sum_{j=1}^\nnl \left[ \left( \lambda_j \1_\nnl^\top \right) \circ \Wnl + \lambda_j (m^j)^\top (I_\nnl - E_j)  \right] \\
    \lambda_j & = -\pinl_j \left( \Pinl - (I_\nnl - E_j) \Qnl^\top \right)^{-1} \pinl.
    \end{aligned}
\] Note that the first term from Eq.~\eqref{eq:grad_liftK} is absent here because $\pinl$ is fixed, whereas $\pil$ is only subject to the lifting constraint $\pinl = V^\top \pil$.} subject to the stationary distribution constraint $\Qnl \1_\nnl = \Qnl^\top \1_\nnl = \pinl$ and graph constraint $\epsilon \mathcal{A} \leq \Qnl \leq \mathcal{A}$ where $\mathcal{A}$ denotes the binary adjacency matrix of the graph. We sweep over lifting budgets which set the total number of virtual states $\nl$. Then we assign virtual states to nodes proportional to the following quantities:
(i) $\1_\nnl$ (uniformly), (ii) stationary distribution $\pinl$, (iii) node degree, (iv) node betweenness centrality, (v) node eigenvector centrality. For these and subsequent experiments, we use Webster's method of apportionment to assign discrete numbers of virtual states proportional to continuous quantities~\cite{MLB-HPY:80}. For each lifting budget and lifting method listed above, we minimize the lifted Kemeny constant of an MC $\Ql$ that is a lifting of $\Qnl$ via the approach described in Section~\ref{sec:lift_opt}. For both the optimized non-lifted and lifted chains, we take the highest-performing solution from 10 random initializations. Fig.~\ref{fig:lifting_approaches} shows boxplots of the distribution of percent improvement in lifted Kemeny constant $\liftK(\Ql)$ over the Kemeny constant $\kemeny(\Qnl)$ for each method and budget. It can be seen that there is a trend of diminishing returns in performance from increasing the budget. We will proceed with the stationary distribution method in subsequent experiments.

\begin{figure}
\centering
\includegraphics[width=\linewidth]{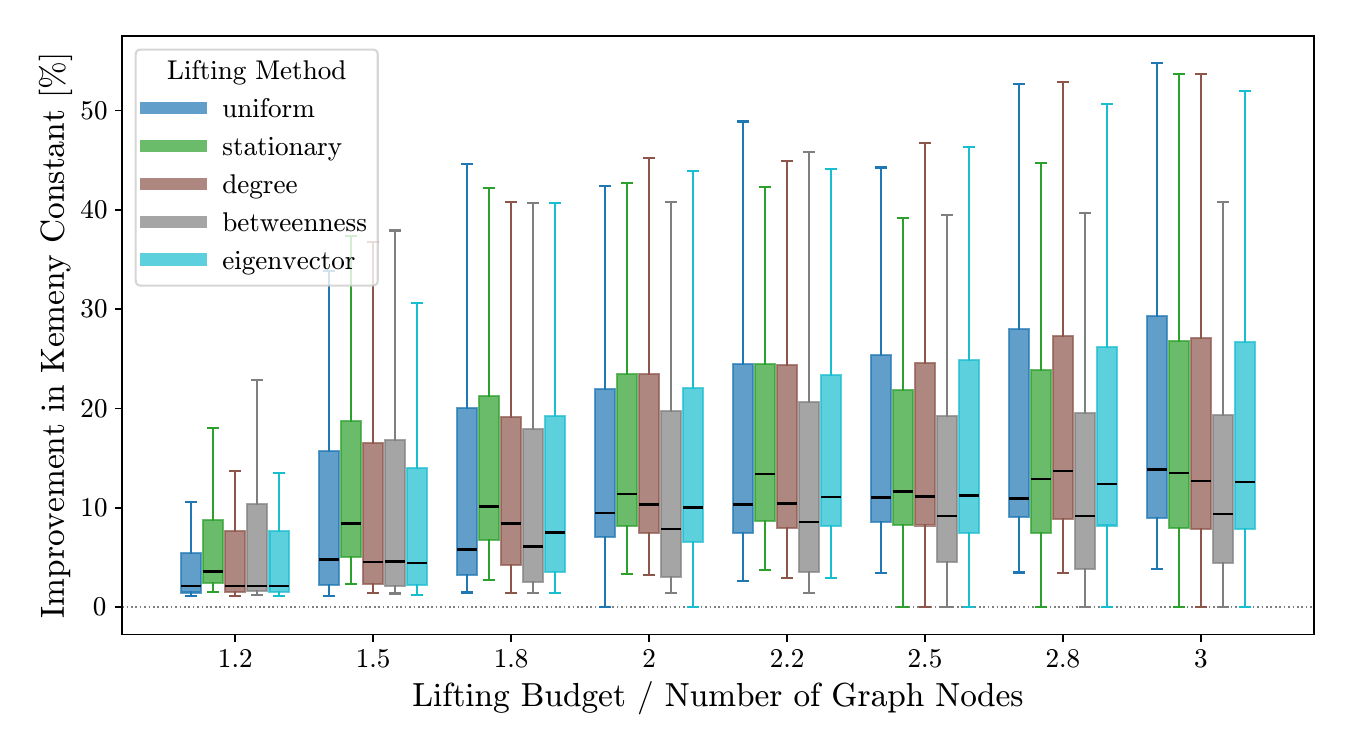}
\caption{Results from the optimization of lifted MC $\Ql$ and non-lifted MC $\Qnl$ patrol strategies on 10-node Erdős--Rényi random undirected graphs conditioned on connectivity. The boxplots show the percent improvement in the optimized $\liftK(\Ql)$ compared to the optimized $\kemeny(\Qnl)$. Each color represents a different lifting method (see legend), and each box represents the distribution of results from 50 Monte Carlo runs.}
\label{fig:lifting_approaches}
\end{figure}

\subsection{General Graphs and Stationary Distributions}\label{sec:p_sweep}
In this section, we evaluate numerically the scope for performance improvement via lifting for the three metrics. We use the same randomly generated graph and stationary distribution approach as in the previous section, but here we generate 50 graphs for each value of the ER edge inclusion probability $p$ shown on the vertical axes of Fig.~\ref{fig:ER_ridgeline}. The left ridgeline plot in Fig.~\ref{fig:ER_ridgeline} shows the distribution of percent decreases of $\liftK(\Ql)$, with a lifting budget of $3\nnl = 30$ and the stationary distribution lifting method, compared to $\kemeny(\Qnl)$ for each value of $p$. It can be seen that substantial performance improvement is possible via lifting. There is also a trend of increasing improvements with graph sparsity, which aligns with the results from the corridor environment.
We perform similar experiments with the (lifted) SGCP metric and the (lifted) RTE metric. The main differences with the SGCP metrics are that there is no stationary distribution constraint on the non-lifted MC $\Qnl$, the metric itself penalizes reducibility so the constraint parameter $\epsilon = 0$, and there is an additional attack duration parameter $\tau \in \natural_{>0}^{\nnl}$. For each random graph, we set $\tau = d\cdot\1_\nnl$ where $d$ is the diameter of the graph, i.e., the shortest attack duration with a nonzero capture probability. See~\cite{XD-DP-FB:19b, YJ-GDG-XD-JRM-FB:23q} for more details.
For the RTE metric, the main difference is the addition of the truncation accuracy parameter $\eta \in (0,1)$ which we set to 0.1 for our experiments. See~\cite{XD-MG-FB:17o} for more details.
Additionally, the objective is to maximize both of these metrics, and we rely on automatic differentiation using JAX instead of closed forms for the gradient. 
The SGCP results are shown in the middle plot of Fig.~\ref{fig:ER_ridgeline} where meaningful increases in the performance of the lifted MC can be seen. There is a similar trend of performance improvement scaling with graph sparsity. Note that for smooth optimization we approximate the $\min$ in the (lifted) SGCP metric with the LogSumExp function and anneal its temperature from $10^{-3}$ to $10^{-6}$. However, we evaluate the optimized MCs via the true (lifted) SGCP metric.
The RTE results are shown in the right plot of Fig.~\ref{fig:ER_ridgeline} where increases in the performance of the lifted MC can be seen. Note that we report the percent increase in Return-Time Perplexity (RTP) $\mathcal{P} = \exp(\mathcal{H})$ rather than return-time entropy $\mathcal{H}$ because entropy is logarithmically scaled. Table~\ref{tab:comp_times} shows the average computation times for the optimization procedure of the Non-Lifted (NL) strategy $\Qnl$ and Lifted (L) strategy $\Ql$ with each metric. As expected, the lifted optimization is generally slower due to the increased number of parameters. However, we restrict attention to small graphs that can be covered effectively by a single patroller; therefore, the optimization of the lifted strategy is feasible. 

\begin{figure*}
    \centering
    \includegraphics[width=\linewidth]{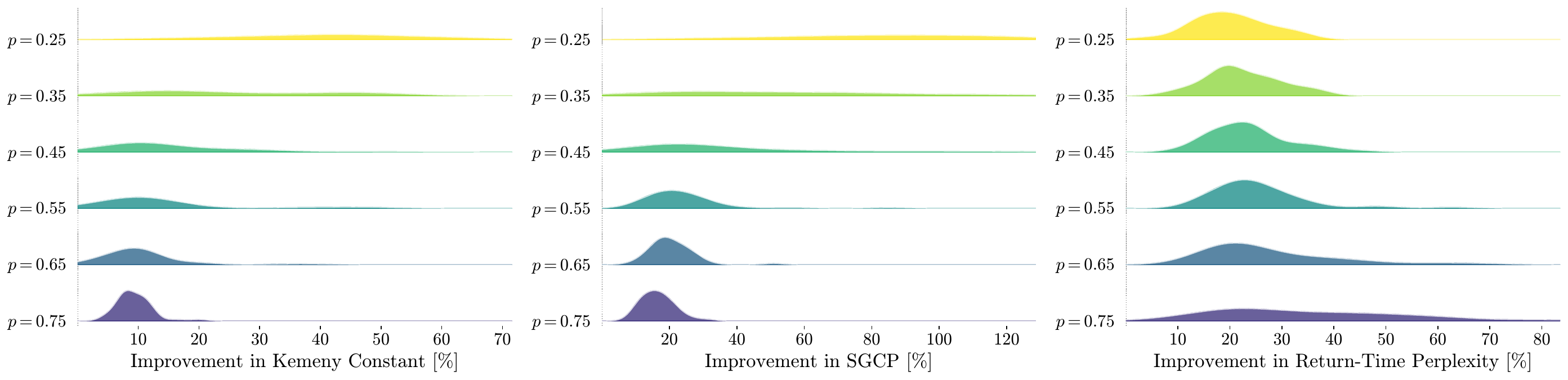}
    \caption{Results from the optimization of lifted and non-lifted MC patrol strategies on 10-node Erdős--Rényi random undirected graphs conditioned on connectivity. The vertical axes list the ER edge inclusion probability $p$. The ridgeline plots show the percent improvement in lifted metric values achieved by the optimized lifted MCs $\Ql$ compared to the metric values achieved by the optimized non-lifted MCs $\Qnl$. Each ridgeline plot represents the distribution of results from 50 Monte Carlo runs.}
    \label{fig:ER_ridgeline}
\end{figure*}

\begin{table}[htbp]
    \centering
    \caption{ER Graph MC Optimization Mean Wall Time Data}
    \label{tab:comp_times}
    \begin{tabular}{| l | c | c | c | c | c | c |} 
        \hline

        & \multicolumn{2}{c|}{Kemeny [s]} & \multicolumn{2}{c|}{SGCP [s]} & \multicolumn{2}{c|}{RTP [s]} \\
        
        \cline{2-7}
        
        & NL & L & NL & L & NL & L \\
        \hline
        
        $p = 0.25$ & $9.1$ & $9.1$ & $6.0$ & $44.1$ & $23.9$ & $47.1$ \\
        $p = 0.35$ & $9.3$ & $9.4$ & $7.1$ & $48.1$ & $25.8$ & $46.5$ \\
        $p = 0.45$ & $5.5$ & $9.8$ & $7.2$ & $49.3$ & $32.1$ & $50.7$ \\
        $p = 0.55$ & $4.0$ & $9.3$ & $4.7$ & $33.9$ & $34.5$ & $55.7$ \\
        $p = 0.65$ & $2.8$ & $9.6$ & $4.1$ & $30.5$ & $27.1$ & $44.2$ \\
        $p = 0.75$ & $2.9$ & $11.9$ & $4.3$ & $30.4$ & $21.5$ & $36.6$ \\
        \hline
    \end{tabular}
\end{table}

\subsection{San Francisco Case Study}
In this section we apply the proposed lifted MC optimization to a real-world example shown in Fig.~\ref{fig:SF_crime}, a 12-node weighted graph corresponding to police patrol locations in San Francisco. See~\cite{SA-EF-SLS:14} for details. We set the stationary distribution proportional to the historic crime rates at each location. Note that Theorem~\ref{thm:bounds} and Corollary~\ref{cor:lmi} were shown for unweighted graphs so the results shown here are purely empirical. 
The results for all three metrics are shown in Table~\ref{tab:SF}. As before, we first optimize a non-lifted MC and then optimize a lifting of that MC. Here we use stationary lifting with a budget of $3\nnl = 36$ virtual states. For the RTE metric we set $\eta = 0.1$. Each row of the table corresponds to a different version of the graph. The original graph is complete with the longest edges having a weight of $9$ minutes. This is the last row of the table. For each previous row we prune edges with weight larger than the given value of $\wnl_\textup{max}$. This is a natural way to decrease the density of the graph because time spent by the patroller traversing edges is not beneficial for the Kemeny or SGCP metrics. For the SGCP metric with all versions of the graph we use the attack duration $\tau = 13\cdot\1_\nnl$, the diameter of the sparsest graph, so that the results can be compared between rows of the table. We can see that the SGCP metric results show a consistent trend of greater improvements with graph sparsity. However, unlike the previous ER graph results, here the Kemeny constant results do not show a clear trend. We believe that this is because the optimized MCs avoid long edges so their removal from the graph has little effect. This is supported by the small increase in Kemeny constant as $\wnl_\textup{max}$ decreases. The takeaway is that this methodology can be applied to real-world graphs, but one should expect larger benefits for sparse topologies such as indoor environments or transportation systems. 
Finally, if a naive adversary records a trajectory of the optimized lifted MC and uses a plug-in estimator to estimate a non-lifted MC, then their estimate would converge to the collapsed MC. This would lead the adversary to expect the performance in the `NL' columns but experience the performance in the `L' columns. This is easily interpretable for the Kemeny and SGCP metrics: they would overestimate the mean time to a stationary-random node from a stationary-random node by $1$ minute, and they would underestimate the SGCP by about $1$ percentage point (see last row of table).

\begin{figure}
\centering
\includegraphics[width=\linewidth]{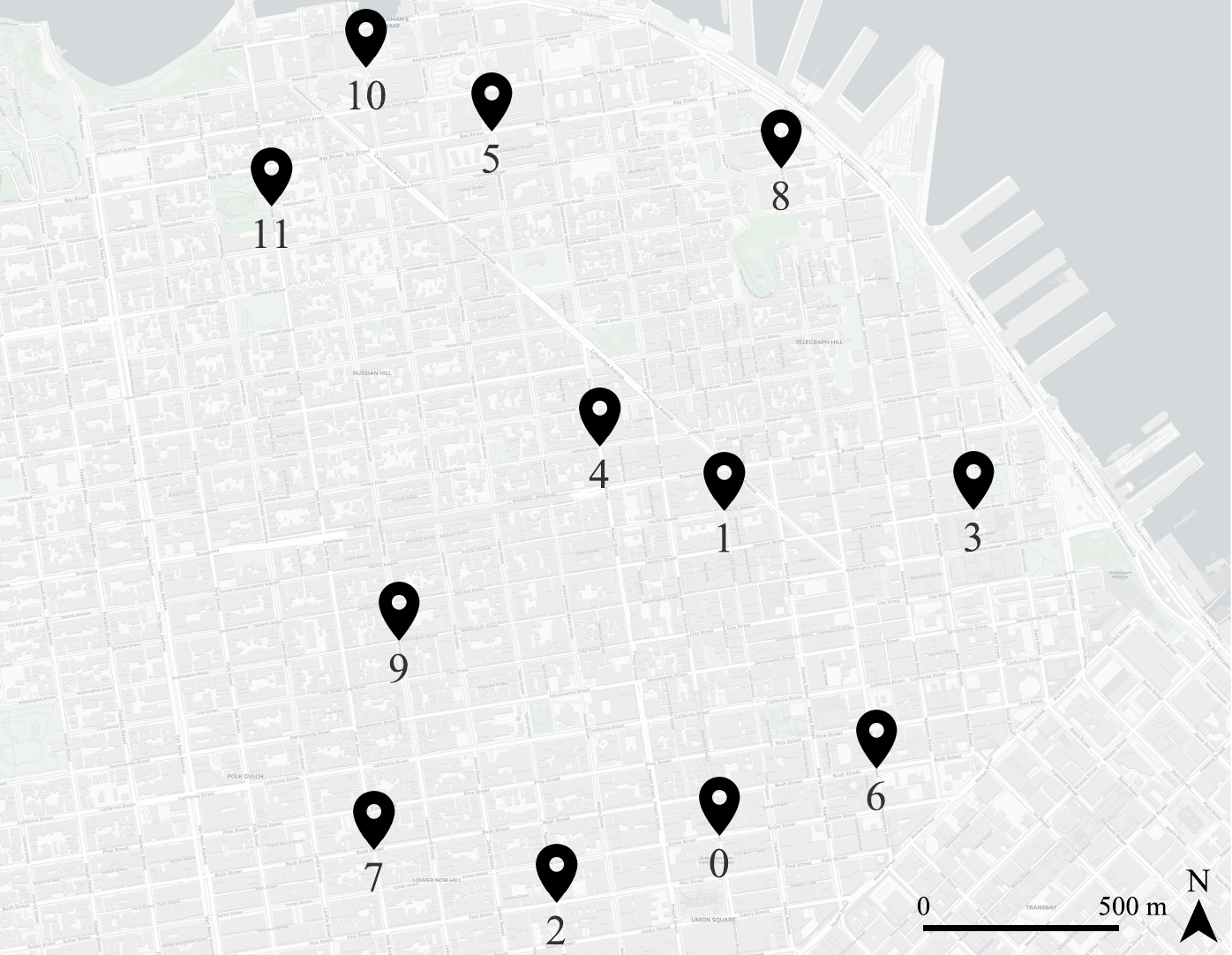}
\caption{A real-world environment consisting of twelve police patrol locations in downtown San Francisco. This benchmark is taken from~\cite{SA-EF-SLS:14}.}
\label{fig:SF_crime}
\end{figure}

\begin{table}[htbp]
    \centering
    \caption{San Francisco Case Study Results}
    \label{tab:SF}
    \setlength{\tabcolsep}{3pt} 
    \begin{tabular}{| c | c | c | c | c | c | c | c | c | c |}
        \hline
        & \multicolumn{3}{c|}{Kemeny [min]} & \multicolumn{3}{c|}{SGCP [\%]} & \multicolumn{3}{c|}{RTP} \\
        \cline{2-10}
        $w_{\max}$ & NL & L & \(\Delta\) & NL & L & \(\Delta\) & NL & L & \(\Delta\) \\
        \hline
        4 & 22.6 & 22.1 & 2.2\% & 7.3  & 9.4  & 28\%  & 57  & 62  & 8.8\% \\
        5 & 22.3 & 21.9 & 1.8\% & 11.1 & 13.1 & 18\%  & 94  & 102 & 8.5\% \\
        6 & 22.3 & 21.6 & 3.1\% & 13.4 & 15.0 & 12\%  & 120 & 130 & 8.3\% \\
        7 & 22.3 & 21.3 & 4.5\% & 14.7 & 15.9 & 8.2\% & 132 & 147 & 11\% \\
        8 & 22.3 & 21.3 & 4.5\% & 15.5 & 16.7 & 7.8\% & 138 & 154 & 12\% \\
        9 & 22.3 & 21.3 & 4.5\% & 17.3 & 18.1 & 4.6\% & 141 & 157 & 11\% \\
        \hline
    \end{tabular}
\end{table}

\subsection{University of Coimbra Case Study}
In this section we apply the proposed lifted MC optimization to another real-world example, a sparse 18-node weighted graph shown in Fig.~\ref{fig:CTCV} corresponding to indoor patrol locations in the CTCV building at the University of Coimbra. This graph is taken from a collection of graph patrolling benchmarks; see~\cite{DP-LI-AF:19} for details. We choose this graph because it is suitable for a single patroller, and we enforce a uniform stationary distribution. We use stationary lifting with a budget of $3\nnl = 54$ virtual states. We assume a patroller speed of \qty{1}{m/s} to compute travel times. For the SGCP metric, we choose the attack duration $\tau = 159\cdot \1_\nnl$ seconds which is three times the diameter to ensure a meaningful capture probability. For the RTE metric we set $\eta = 0.1$. The results are given in Table~\ref{tab:CTCV}. A significant performance improvement can be seen, particularly for the Kemeny and SGCP metrics. This aligns with previous results which indicate that sparse graphs afford the greatest improvements via lifting. As before, the naive adversary would expect the performance in the `NL' columns but experience the performance in the `L' columns. 

\begin{figure*}
\centering
\includegraphics[width=\linewidth]{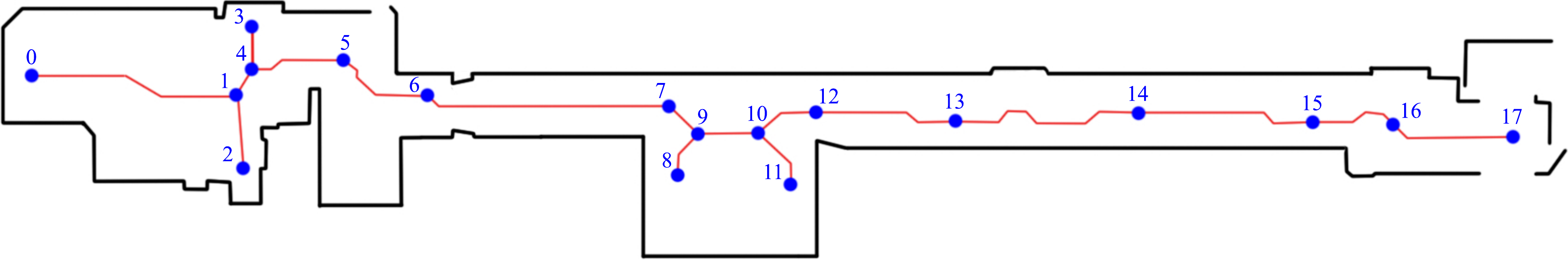}
\caption{A real-world environment of eighteen patrol locations in the CTCV building at the University of Coimbra. This benchmark is taken from~\cite{DP-LI-AF:19}.}
\label{fig:CTCV}
\end{figure*}

\begin{table}[htbp]
    \centering
    \caption{University of Coimbra Case Study Results}
    \label{tab:CTCV}
    \begin{tabular}{| c | c | c | c | c | c | c | c | c |} 
        \hline

        \multicolumn{3}{|c|}{Kemeny [s]} & \multicolumn{3}{c|}{SGCP [\%]} & \multicolumn{3}{c|}{RTP} \\
        
        \cline{1-9}
        
        NL & L & \(\Delta\) & NL & L & \(\Delta\) & NL & L & \(\Delta\) \\
        \hline
        
        315 & 72 & 77\% & 6 & 26 & 333\% & 9 & 10 & 11\% \\
        \hline
    \end{tabular}
\end{table}

\section{Conclusion}
In this paper we propose the use of classical MC lifting for improving the performance of stochastic patrol strategies. We design generalizations of the Kemeny constant, SGCP, and RTE metrics from the literature to appropriately measure the performance of lifted MCs. We prove theoretical bounds on the performance of the lifted MCs and demonstrate the novel deceptive quality of lifted MCs against model-based adversaries. Finally, we show numerically that substantial performance improvement is possible on sparse graphs. This work suggests several topics for future study, including extensions of the theoretical results to accommodate travel times, the optimal assignment of virtual states to nodes, and approaches for the adversary against lifted MCs.

\appendix

\subsection{Lumpability}\label{app:lumpability}
Given an MC $\Pl$ and partition $V$ we can always use Eq.~\eqref{eq:coll} to find the collapsed MC $\Pnl$. However, depending on $\Pl,V$, information may be lost in this process, and $\Pnl$ may not exactly reproduce the behavior of $\Pl$ on the graph. \emph{Lumpability} is the property of $\Pl,V$ that determines whether or not information loss occurs. Consider the state space $\Sl = \{1,\dots,\nl\}$ and partition $\Snl = \{S_1,\dots,S_\nnl\}$. We encode the partition $\Snl$ using an indicator matrix $V$. Given an MC $\Pl$ with values $\Xl_k \in \Sl$ and initial distribution $\pl_0$, we can characterize the stochastic process with values $\Xnl_k \in \Snl$ as $\Pr[\Xnl_0 = S_i] = \Pr_{\pl_0}[\Xl_0 \in S_i]$ with transition probabilities 
\[
    \Pr[\Xnl_1 = S_j ~|~ \Xnl_0 = S_i] = \Pr_{\pl_0}[\Xl_1 \in S_j ~|~ \Xl_0 \in S_i].
\]
Generalizing to the $k$-th step, we have that
\begin{equation*}
\begin{aligned}
    & \Pr \left[\Xnl_k = S_t ~|~ \Xnl_{k-1} = S_s, \dots,\Xnl_1 = S_j, \Xnl_0 = S_i \right] \\
    & = \Pr \left[ \Xl_k \in S_t ~|~ \Xl_{k-1} \in S_s, \dots, \Xl_1 \in S_j, \Xl_0 \in S_i \right].
\end{aligned}
\end{equation*}
We call $\Xnl_k$ the lumped chain of $\Xl_k$ with respect to $\Snl$. The above conditions ensure that the lumped chain exactly captures the behavior of $\Xl_k$ on $\Snl$. However, the lumped chain may not be an MC, motivating the following~\cite[Def.~6.3.1]{JGK-JLS:76}:
\begin{defn}[Lumpability]~\label{def:lump}
    An MC on state space $\Sl$ is \emph{lumpable} with respect to a partition $\Snl$ if, for every starting distribution $\pl_0$, the lumped chain is an MC and its transition probabilities do not depend on $\pl_0$.
\end{defn}
Equivalently, we can state that an MC $\Pl$ is \textit{lumpable with respect to} $V$ if and only if it satisfies\footnote{The choice of diagonal positive-definite weighting matrix in $U$ is arbitrary~\cite{JGK-JLS:76}. We choose $\Pil$ for consistency with the definition of lifting.} $VU \Pl V = \Pl V$~\cite[Thm.~6.3.4]{JGK-JLS:76}.
Comparing with Eq.~\eqref{eq:coll} yields
\begin{equation}\label{eq:lump}
    V \Pnl = \Pl V.
\end{equation} 
Now we show that lumpability implies that $\Pl$ is a lifting of $\Pnl$. Pre-multiplying Eq.~\eqref{eq:lump} by $U$ and recognizing that $UV = I_\nnl$ yields the second condition of Eq.~\eqref{eq:coll}. To show the first condition, check that $V^\top \pil$ is the stationary distribution of $\Pnl$:
\[
\begin{aligned}
    \left( V^\top \pil \right)^\top \Pnl & = \left(\pil\right)^\top V \Pnl = \left(\pil\right)^\top \Pl V \\
    & = \left(\pil\right)^\top V = \left( V^\top \pil \right)^\top.
\end{aligned}
\]

We have an example of lumpability in Fig.~\ref{fig:motivating_example}. It can be verified that for $\pl = 1/2$ the lifted MC $\Pl$ is lumpable with respect to $V$, and $\Pnl$ is the lumped MC. 

\subsection{Proof of Lemma~\ref{lem:lump}}\label{app:lem_3}
\begin{proof}
    From the condition for lumpability $V\Pnl = \Pl V$, we have that $\sum_{s\in S_j} \pl_{rs} = \pnl_{ij}$ for all $r \in S_i$. Therefore, paths in $\Sl$ correspond to paths in $\Snl$, and the irreducibility of $\Pl$ implies the irreducibility of $\Pnl$. Irreducibility in turn implies finite mean first passage times in $M,M^\textup{lift}$. 
    
    We proceed by induction to prove Eq.~\eqref{eq:VF}. The base case $k = 1$ is trivial. We have that $F_1^\textup{lift} = \Pl V, F_1 = \Pnl$. Substituting into the lumpability condition $V\Pnl = \Pl V$ yields $VF_1 = F_1^\textup{lift}$. Now assume $F_k^\textup{lift} = VF_k$. We want to show that $F_{k+1}^\textup{lift} = VF_{k+1}$. From Eq.~\eqref{eq:F_k}, we have that $F_{k+1} = \Pnl[F_k - \textup{diag}(F_k)]$. Pre-multiplying by V and using lumpability yields
    \[
        VF_{k+1} = V\Pnl [F_k - \textup{diag}(F_k)] = \Pl [VF_k - V\textup{diag}(F_k)].
    \]
    From Eq.~\eqref{eq:F_k_lift} and the inductive hypothesis, we have that
    \[
        F_{k+1}^\textup{lift} = \Pl[F_k^\textup{lift} - (F_k^\textup{lift} \circ V)] = \Pl [VF_k - ((VF_k) \circ V)].
    \]
    It can be shown due to the structure of $V$ that $V\textup{diag}(F_k) = (VF_k) \circ V$ completing the proof of Eq.~\eqref{eq:VF}. Then Eq.~\eqref{eq:VM} follows using the finiteness of the mean first passage times:
    \[
        M^\textup{lift} = \sum_{k=1}^\infty kF_k^\textup{lift} = \sum_{k=1}^\infty k(VF_k) = V \sum_{k=1}^\infty kF_k = VM.
    \]
\end{proof}

\subsection{Proof of Lemma~\ref{lem:m_Q_rel}}\label{app:lem_4}
\begin{proof}
    Define the vector $m_A^{A^c} \in \real^{|A|}$ containing entries $m_i^{A^c}$ for all $i \in A$. $m_A^{A^c}$ can be shown to satisfy $m_A^{A^c} = (I_{|A|} - P_A)^{-1} \1_{|A|}$   where $P_A$ is the principal submatrix of $P$ restricted to the states in $A$. Because $P_A$ has a spectral radius less than one, $(I_{|A|} - P_A)^{-1} = \sum_{k=0}^\infty P_A^k$. Therefore, we can write
    \begin{equation}\label{eq:c_seq_sum}
    \sum_{i \in A} \pi_i m_i^{A^c} = \pi_A^\top m_A^{A^c} = \sum_{k=0}^\infty c_k,
    \end{equation}
    where $c_k = \pi_A^\top P_A^k \1_{|A|}$ and $\pi_A \in \real^{|A|}$ is the stationary distribution restricted to the states in $A$. Note that $c_0 = \pi(A)$ and
    \[
    c_0 - c_1 = \pi_A^\top (\1_{|A|} - P_A \1_{|A|}) = \sum_{i \in A} \pi_i \sum_{j \in A^c} p_{ij} = Q(A,A^c).
    \]
    We proceed by showing that the decreasing sequence $c_k$ is convex. Define the increments $d_k = c_k - c_{k+1}$; convexity of $c_k$ follows from showing that the sequence $d_k$ is decreasing:
    \[
    \begin{aligned}
    d_k & = \pi_A^\top P_A^k (I_{|A|} - P_A) \1_{|A|} {=} \left( \pi_A^\top P_A \right) P_A^{k-1}(I_{|A|} - P_A) \1_{|A|} \\
    & \leq \pi_A^\top P_A^{k-1} (I_{|A|} - P_A) \1_{|A|} = d_{k-1},
    \end{aligned}
    \]
    where the inequality arises because $\pi_A^\top P_A \leq \pi_A^\top$. 
    The infinite sum in Eq.~\eqref{eq:c_seq_sum} can be viewed as a left Riemann sum. Now consider extending the secant line from $c_0$ to $c_1$ down to the abscissa. This defines a triangle between the axes and the secant line. Due to convexity and the fact that $c_k \geq 0$ for all $k$, the area of this triangle lower bounds the left Riemann sum. This triangle has side lengths $\pi(A)$ and $\pi(A)/Q(A,A^c)$. Thus we have the bound given in Eq.~\eqref{eq:m_Q_rel}.
\end{proof}

\subsection{Proof of Proposition~\ref{prop:deception}}\label{app:deception}
\begin{proof}
    By Definition~\ref{def:plugin}, the estimated transition probabilities between nodes $i,j \in \Snl$ satisfy
    \[
    \begin{aligned}
    \lim_{T\to\infty} \hat{p}_{ij}(\xnl_T) & =  \lim_{T\to\infty} \frac{n_{ij}(\xnl_T)}{n_i(\xnl_T)} = \lim_{T\to\infty} \frac{n_{S_i S_j}(\xl_T)}{n_{S_i}(\xl_T)} \\ 
    & = \lim_{T\to\infty} \frac{\sum_{k\in S_i}\sum_{l\in S_j} n_{kl}(\xl_T)}{\sum_{k\in S_i} n_k(\xl_T)},
    \end{aligned}
    \]
    where $n_{kl}(\xl_T)$ are the empirical transition counts between virtual states $k,l$ and $n_k(\xl_T)$ is the empirical visit count to virtual state $k$. Dividing the numerator and denominator by $T$ allows the application of Eq.~\eqref{eq:stat_dist}. The denominator becomes
    \[
    \lim_{T\to\infty} \frac{\sum_{k\in S_i} n_k(\xl_T)}{T} = \sum_{k\in S_i} \pil_k = \pil(S_i), \quad \textup{a.s.}
    \]
    Now for the numerator consider an MC $\Plp$ on state space $\Slp = \{(k,l) \in \Sl \times\Sl\ |\ \pl_{kl} >0 \}$ whose transition probabilities are $\plp_{(k,l)(l',r)} = \pl_{lr}$ if $l = l'$ and zero otherwise. It can be verified that $\Plp$ is also irreducible and has stationary distribution $\pilp_{kl} = \pil_k \pl_{kl}$. Applying Eq.~\eqref{eq:stat_dist} via $\Plp$ yields
    \[
    \begin{aligned}
    \lim_{T\to\infty} \frac{\sum_{k\in S_i}\sum_{l\in S_j} n_{kl}(\xl_T)}{T} & = \sum_{k\in S_i} \sum_{l\in S_j} \pilp_{kl} \\ 
    & = \sum_{k\in S_i} \pil_k \sum_{l\in S_j} \pl_{kl}, \quad \textup{a.s.}
    \end{aligned}
    \]
    Then the ratio is identical to the transition probability condition for lifting given in Eq.~\eqref{eq:coll_elementwise} which in matrix form yields Eq.~\eqref{eq:deception} where $\Pnl = U \Pl V$.
\end{proof}

\subsection{Proof of Lemma~\ref{lem:mean_return_time}}\label{app:return_time}
\begin{proof}
    If we can show that $\sum_{l\in S_j} \pil_l m_l^{S_j} = \pinl^\top (\Pnl \circ \Wnl) \1_\nnl$, then the rest of the argument is the same as in Section~\ref{sec:lift_RTE}.
    Rearranging Eq.~\eqref{eq:liftM} and pre-multiplying by the transpose of $\pil$ yields 
    \[
        \left(\pil\right)^\top D_j m^{S_j} = \left(\pil\right)^\top (\Pl\circ\Wl) \1_{\nl}. 
    \]
    Then due to the structure of $D_j$, the LHS is equivalent to $\sum_{l\in S_j} \pil_l m_l^{S_j}$. The RHS can be rewritten
    \[
        \sum_{k,l=1}^{\nl} \ql_{kl} \wl_{kl} = \sum_{i,j=1}^\nnl \wnl_{ij} \sum_{k\in S_i} \sum_{l\in S_j} \ql_{kl}  = \sum_{i,j=1}^\nnl \qnl_{ij} \wnl_{ij},
    \]
    where we have used $\Wl = V \Wnl V^\top$ and $\Qnl = V^\top \Ql V$. The final expression is equivalent to $\pinl^\top (\Pnl\circ\Wnl) \1_\nnl$.
\end{proof}

\subsection{Proof of Proposition~\ref{prop:grad_lift_K}}\label{app:gradient}
\begin{proof}
    First, we will aim to introduce $\Ql$ in Eq.~\eqref{eq:liftK} and Eq.~\eqref{eq:liftM} and eliminate $\Pl,\pil$. We can rewrite Eq.~\eqref{eq:liftK} as
    \[
    \liftK = \sum_{j=1}^\nnl \pinl_j \left[ \left(\pil\right)^\top m^{S_j} \right] = \sum_{j=1}^\nnl \pinl_j \left[ \1_{\nl}^\top \left(\Ql\right)^\top m^{S_j} \right].
    \]
    Pre-multiplying Eq.~\eqref{eq:liftM} by $\Pil = \textup{diag}(\Ql\1_{\nl})$ and rearranging yields
    \[
    \left[\textup{diag}(\Ql\1_{\nl}) {-} \Ql(I_{\nl}-D_j)\right] m^{S_j} - (\Ql \circ \Wl)\1_{\nl} = 0
    \]
    whose LHS we will denote $C_j(\Ql,m^{S_j})$. We proceed by the adjoint state method~\cite{JC:86}. Define the Lagrangian 
    \[
    \mathscr{L} = \sum_{j=1}^\nnl \left[ \pinl_j \1_{\nl}^\top \left(\Ql\right)^\top m^{S_j} + \lambda_j^\top C_j(\Ql,m^{S_j}) \right].
    \]
    Differentiating with respect to $m^{S_j}$ and setting equal to zero we arrive at an equation for the adjoint state variables
    \[
    \lambda_j = -\pinl_j \left[ \textup{diag}(\Ql\1_{\nl}) {-} (I_{\nl} - D_j) \left(\Ql\right)^\top \right]^{-1} \Ql \1_{\nl},
    \]
    which can be simplified to Eq.~\eqref{eq:lambda}. Finally, differentiating $\mathscr{L}$ with respect to $\Ql$ yields Eq.~\eqref{eq:grad_liftK}.
\end{proof}

\bibliographystyle{IEEEtran}
\bibliography{bib/alias,bib/FB,bib/Main,bib/New,bib/YJ_TAC}

\begin{thebibliography}{10}
\providecommand{\url}[1]{#1}
\csname url@samestyle\endcsname
\providecommand{\newblock}{\relax}
\providecommand{\bibinfo}[2]{#2}
\providecommand{\BIBentrySTDinterwordspacing}{\spaceskip=0pt\relax}
\providecommand{\BIBentryALTinterwordstretchfactor}{4}
\providecommand{\BIBentryALTinterwordspacing}{\spaceskip=\fontdimen2\font plus
\BIBentryALTinterwordstretchfactor\fontdimen3\font minus \fontdimen4\font\relax}
\providecommand{\BIBforeignlanguage}[2]{{%
\expandafter\ifx\csname l@#1\endcsname\relax
\typeout{** WARNING: IEEEtran.bst: No hyphenation pattern has been}%
\typeout{** loaded for the language `#1'. Using the pattern for}%
\typeout{** the default language instead.}%
\else
\language=\csname l@#1\endcsname
\fi
#2}}
\providecommand{\BIBdecl}{\relax}
\BIBdecl

\bibitem{NB:22}
N.~Basilico, ``Recent trends in robotic patrolling,'' \emph{Current Robotics Reports}, vol.~3, no.~2, pp. 65--76, 2022.

\bibitem{XD-FB:20a}
X.~Duan and F.~Bullo, ``Markov chain-based stochastic strategies for robotic surveillance,'' \emph{Annual Review of Control, Robotics, and Autonomous Systems}, vol.~4, pp. 243--264, 2021.

\bibitem{RP-PA-FB:14b}
R.~Patel, P.~Agharkar, and F.~Bullo, ``Robotic surveillance and {M}arkov chains with minimal weighted {K}emeny constant,'' \emph{IEEE Transactions on Automatic Control}, vol.~60, no.~12, pp. 3156--3167, 2015.

\bibitem{PD-SH-RMN:00}
P.~Diaconis, S.~Holmes, and R.~M. Neal, ``Analysis of a nonreversible {M}arkov chain sampler,'' \emph{Annals of Applied Probability}, vol.~10, no.~3, pp. 726--752, 2000.

\bibitem{XD-WW-RY:26}
X.~Duan, W.~Wang, and R.~Yan, ``Maximizing {M}arkov trajectory entropy under {K}emeny constraints for robotic surveillance,'' \emph{IEEE Transactions on Automatic Control}, vol.~71, no.~7, pp. 4850--4857, 2026.

\bibitem{BA-DSB:12}
B.~Acikmese and D.~S. Bayard, ``A {M}arkov chain approach to probabilistic swarm guidance,'' in \emph{American Control Conference}, 2012, p. 6300–6307.

\bibitem{GDG-FB-JRM:22v}
G.~D\'iaz-Garc\'ia, F.~Bullo, and J.~R. Marden, ``Distributed {Markov} chain-based strategies for multi-agent robotic surveillance,'' \emph{IEEE Control Systems Letters}, vol.~7, pp. 2527--2532, 2023.

\bibitem{CPCF-AZ-BFH:26}
C.~P.~C. Franssen, A.~Zocca, and B.~F. Heidergott, ``A first-order gradient approach for the connectivity optimization of {M}arkov chains,'' \emph{IEEE Transactions on Automatic Control}, vol.~71, no.~2, pp. 839--854, 2026.

\bibitem{ABA-SS-SLS:25}
A.~B. Asghar, S.~Sundaram, and S.~L. Smith, ``Multirobot persistent monitoring: Minimizing latency and number of robots with recharging constraints,'' \emph{IEEE Transactions on Robotics}, vol.~41, pp. 236--252, 2025.

\bibitem{XL-CGC:15}
X.~Lin and C.~G. Cassandras, ``An optimal control approach to the multi-agent persistent monitoring problem in two-dimensional spaces,'' \emph{IEEE Transactions on Automatic Control}, vol.~60, no.~6, p. 1659–1664, 2015.

\bibitem{DP-RPR:16}
D.~Portugal and R.~P. Rocha, ``Cooperative multi-robot patrol with {Bayesian} learning,'' \emph{Autonomous Robots}, vol.~40, no.~5, pp. 929--953, 2016.

\bibitem{DP-RPR:10}
------, ``{MSP} algorithm: Multi-robot patrolling based on territory allocation using balanced graph partitioning,'' in \emph{ACM Symposium on Applied Computing}, 2010, pp. 1271--1276.

\bibitem{SA-AS-FT:21}
S.~Apers, A.~Sarlette, and F.~Ticozzi, ``Characterizing limits and opportunities in speeding up {M}arkov chain mixing,'' \emph{Stochastic Processes and their Applications}, vol. 136, p. 145–191, 2021.

\bibitem{AN-MTG-MW-BR-MR-TPM:24}
A.~Nag{\'o}rko, M.~Waniek, M.~R{\'o}g, M.~T. Godziszewski, B.~Rosiak, and T.~P. Michalak, ``General {M}arkov model for solving patrolling games,'' in \emph{Conf.\ on Uncertainty in Artificial Intelligence}, 2024, pp. 2646--2669.

\bibitem{PM-SS:17}
P.~Masters and S.~Sardina, ``Deceptive path-planning,'' in \emph{Int.\ Joint Conf.\ on Artificial Intelligence}, 2017, pp. 4368--4375.

\bibitem{NT-NA:17}
N.~Talmor and N.~Agmon, ``On the power and limitations of deception in multi-robot adversarial patrolling,'' in \emph{Int.\ Joint Conf.\ on Artificial Intelligence}, 2017, pp. 430--436.

\bibitem{MOK-MO-UT:23}
M.~O. Karabag, M.~Ornik, and U.~Topcu, ``Exploiting partial observability for optimal deception,'' \emph{IEEE Transactions on Automatic Control}, vol.~68, no.~7, pp. 4443--4450, 2023.

\bibitem{AS-CNH:14}
A.~Saboori and C.~N. Hadjicostis, ``Current-state opacity formulations in probabilistic finite automata,'' \emph{IEEE Transactions on Automatic Control}, vol.~59, no.~1, pp. 120--133, 2014.

\bibitem{DA-JAF:02}
D.~Aldous and J.~A. Fill, \emph{Reversible Markov Chains and Random Walks on Graphs}, 2002, unfinished monograph, recompiled 2014.

\bibitem{EC:13}
E.~\c{C}inlar, \emph{Introduction to Stochastic Processes}.\hskip 1em plus 0.5em minus 0.4em\relax Dover Publications, 2013.

\bibitem{JGK-JLS:76}
J.~G. Kemeny and J.~L. Snell, \emph{Finite Markov Chains}.\hskip 1em plus 0.5em minus 0.4em\relax Springer, 1976.

\bibitem{XD-DP-FB:19b}
X.~Duan, D.~Paccagnan, and F.~Bullo, ``Stochastic strategies for robotic surveillance as {Stackelberg} games,'' \emph{IEEE Transactions on Control of Network Systems}, vol.~8, no.~2, pp. 769--780, 2021.

\bibitem{YJ-GDG-XD-JRM-FB:23q}
Y.~John, G.~D\'iaz-Garc\'ia, X.~Duan, J.~R. Marden, and F.~Bullo, ``A stochastic surveillance {Stackelberg} game: Co-optimizing defense placement and patrol strategy,'' \emph{IEEE Transactions on Automatic Control}, vol.~70, no.~8, pp. 5468--5474, 2025.

\bibitem{XD-MG-FB:17o}
X.~Duan, M.~George, and F.~Bullo, ``Markov chains with maximum return time entropy for robotic surveillance,'' \emph{IEEE Transactions on Automatic Control}, vol.~65, no.~1, pp. 72--86, 2020.

\bibitem{FC-LL-IP:99}
F.~Chen, L.~Lov{\'a}sz, and I.~Pak, ``Lifting {M}arkov chains to speed up mixing,'' in \emph{Annual ACM Symposium on Theory of Computing}, May 1999, pp. 275--281.

\bibitem{DaL-YP:17}
D.~A. Levin and Y.~Peres, \emph{Markov Chains and Mixing Times}, 2nd~ed.\hskip 1em plus 0.5em minus 0.4em\relax American Mathematical Society, 2017.

\bibitem{LL:78b}
L.~Ljung, ``Convergence analysis of parametric identification methods,'' \emph{IEEE Transactions on Automatic Control}, vol.~23, no.~5, p. 770–783, 1978.

\bibitem{HW:82}
H.~White, ``Maximum likelihood estimation of misspecified models,'' \emph{Econometrica}, vol.~50, no.~1, pp. 1--25, 1982.

\bibitem{KD-PGM-SPM:11}
K.~Deng, P.~G. Mehta, and S.~P. Meyn, ``Optimal {K}ullback-{L}eibler aggregation via spectral theory of {M}arkov chains,'' \emph{IEEE Transactions on Automatic Control}, vol.~56, no.~12, p. 2793–2808, 2011.

\bibitem{BCG-TP-GK-HK:15}
B.~C. Geiger, T.~Petrov, G.~Kubin, and H.~Koeppl, ``Optimal {K}ullback–{L}eibler aggregation via information bottleneck,'' \emph{IEEE Transactions on Automatic Control}, vol.~60, no.~4, pp. 1010--1022, 2015.

\bibitem{YO:21}
Y.~Ohta, ``On the realization of hidden {M}arkov models and tensor decomposition,'' \emph{IFAC-PapersOnLine}, vol.~54, no.~9, pp. 725--730, 2021.

\bibitem{KJ-DS-JS:07}
K.~Jung, D.~Shah, and J.~Shin, ``\BIBforeignlanguage{English}{Fast and slim lifted {M}arkov chains},'' in \emph{\BIBforeignlanguage{English}{Allerton Conf.\ on Communication, Control and Computing}}, 2007, pp. 1239--1246.

\bibitem{MLB-HPY:80}
M.~L. Balinski and H.~P. Young, ``The {W}ebster method of apportionment,'' \emph{Proceedings of the National Academy of Sciences}, vol.~77, no.~1, pp. 1--4, 1980.

\bibitem{SA-EF-SLS:14}
S.~Alamdari, E.~Fata, and S.~L. Smith, ``Persistent monitoring in discrete environments: Minimizing the maximum weighted latency between observations,'' \emph{International Journal of Robotics Research}, vol.~33, no.~1, pp. 138--154, 2014.

\bibitem{DP-LI-AF:19}
D.~Portugal, L.~Iocchi, and A.~Farinelli, \emph{A {ROS}-Based Framework for Simulation and Benchmarking of Multi-robot Patrolling Algorithms}.\hskip 1em plus 0.5em minus 0.4em\relax Springer, 2019, pp. 3--28.

\bibitem{JC:86}
J.~C{\'e}a, ``Conception optimale ou identification de formes, calcul rapide de la d{\'e}riv{\'e}e directionnelle de la fonction co{\^u}t,'' \emph{ESAIM: Mod{\'e}lisation math{\'e}matique et analyse num{\'e}rique}, vol.~20, no.~3, pp. 371--402, 1986.

\end{thebibliography}

\vfill

\end{document}